\documentclass[a4paper,UKenglish]{lipics-v2021}

\usepackage[utf8]{inputenc}
\usepackage{xspace}
\usepackage{amsmath}
\usepackage{amsthm}
\usepackage{amsfonts}
\usepackage{amssymb}
\usepackage{mathtools}
\usepackage{textcomp}

\usepackage{varwidth}
\usepackage{graphicx}
\usepackage{mathrsfs} 
\usepackage{subfiles}
\usepackage[noadjust]{cite}

\usepackage{tikz}
\usetikzlibrary{automata, graphs,positioning,chains,arrows}

\let\oldrestriction\restriction
\renewcommand{\restriction}{\!\oldrestriction}

\newcommand{\M}{\mathcal{M}}

\renewcommand{\paragraph}[1]{\smallskip
        \noindent\textbf{#1}.}

\newcommand{\strucA}{\mathcal A}
\newcommand{\strucB}{\mathcal B}

\newcommand{\strucD}{\mathcal D}

\newcommand{\graphG}{\mathcal G}

\newcommand{\graphH}{\mathcal{H}}
\newcommand{\graphK}{\mathcal K}

\newcommand{\treeT}{\mathcal T}

\newcommand{\vertexG}{V_{\graphG}}
\newcommand{\vertexH}{V_{\graphH}}
\newcommand{\edgeG}{E_{\graphG}}
\newcommand{\edgeH}{E_{\graphH}}
\newcommand{\vertexT}{V_{\treeT}}

\newcommand{\classA}{\mathsf A}
\newcommand{\classB}{\mathsf B}

\newcommand{\classT}{\mathsf T}

\newcommand{\Hom}{\ensuremath{\operatorname{Hom}}}

\newcommand{\classall}{\textbf{\_}}

\newcommand{\dom}{\operatorname{\mathsf{dom}}}

\newcommand{\weight}{\operatorname{\mathsf{weight}}}
\newcommand{\CSP}{\operatorname{\mathsf{HOM}}}%
\newcommand{\poly}{\operatorname{\mathsf{poly}}}

\newcommand{\tw}{\operatorname{\mathsf{tw}}}
\newcommand{\fhtw}{\operatorname{\mathsf{fhtw}}}
\newcommand{\subw}{\operatorname{\mathsf{subw}}}

\newcommand{\DNNF}{\textsf{DNNF}}

\newcommand{\var}{\operatorname{\mathsf{var}}}

\newcommand{\fec}{\rho^{\ast}}
\newcommand{\C}{\mathtt{C}}

\newcommand{\ctc}{$\{\cup,\times\}$-circuit}
\newcommand{\R}{\mathcal{R}}

\newcommand{\strucX}{\mathcal X}
\newcommand{\strucY}{\mathcal Y}

\title{Factorised Representations of Join Queries: Tight Bounds and a New Dichotomy} 
\titlerunning{Factorised Representations of Join Queries}

\author{Christoph Berkholz}{Technische Universit\"at Ilmenau, Germany}{christoph.berkholz@tu-ilmenau.de}{}{}
\author{Harry Vinall-Smeeth}{Technische Universit\"at Ilmenau, Germany}{harry.vinall-smeeth@tu-ilmenau.de}{}{}

\authorrunning{Christoph Berkholz and Harry Vinall-Smeeth}
\Copyright{Christoph Berkholz and Harry Vinall-Smeeth}

\ccsdesc[500]{Theory of computation~Database theory}

\keywords{join queries,
homomorphisms,
factorised databases,
succinct representation,
knowledge compilation,
lower bounds}

\nolinenumbers %

\hideLIPIcs  %

\funding{This work was funded by the Deutsche Forschungsgemeinschaft (DFG, German Research Foundation) project number 414325841.}

\newcommand{\confORfull}[2]{#2} 

\begin{document}
        
        \maketitle

\begin{abstract}
A common theme in \emph{factorised databases} and \emph{knowledge
  compilation} is the representation of solution sets in a
useful yet succinct data structure. In this paper, we study the
representation of the result of join queries (or,
equivalently, the set of homomorphisms between two relational structures).
We focus on the very general format of $\{\cup,
\times\}$-circuits---also known as \emph{d-representations} or {\DNNF}
circuits---and aim to find the limits of this approach.

In prior work, it has been shown that there always exists a $\{\cup, \times\}$-circuit of
size $N^{O(\subw)}$ representing the query result, where $N$ is the size of the
database and $\subw$ the submodular width of the query.
If the arity of all relations is bounded by a constant, then $\subw$ is linear in the treewidth $\tw$ of the query. In this setting, \confORfull{Berkholz and Vinall-Smeeth}{the authors of this paper} proved a lower bound of
$N^{\Omega(\tw^\varepsilon)}$ on the circuit
size (ICALP 2023), where $\varepsilon>0$ depends on the excluded grid theorem.

Our first main contribution is to improve this lower bound to
$N^{\Omega(\tw)}$, which is tight up to a constant factor
in the exponent. Our second contribution is a
$N^{\Omega(\subw^{1/4})}$ lower bound on the circuit size for
join queries
over relations of unbounded arity.
Both lower bounds are unconditional lower bounds on the circuit size
for well-chosen database instances.
Their proofs use a combination of structural (hyper)graph theory with communication complexity in a simple yet novel way. While the second lower bound is asymptotically equivalent to
Marx's conditional bound on the decision complexity (JACM 2013),
our $N^{\Theta(\tw)}$ bound in the bounded-arity setting is tight, while 
the best conditional bound on the decision complexity is
$N^{\Omega(\tw/\log \tw)}$. Note that, removing this logarithmic factor in the decision setting is a major open problem.
\end{abstract}

        \section{Introduction}\label{sec:introduction}

A central topic in \emph{factorised databases} is the succinct
representation of query results in a data structure that enables
efficient access, ranging from counting, sampling or enumerating the result
tuples to complex analytical tasks. We study the question: \emph{how
  succinctly can the result of a join query be represented?} Of course, the answer depends on the
representation format under consideration. For the flat
representation (listing all result tuples), the answer is given
by the\confORfull{}{ well-known} AGM-bound~\cite{DBLP:conf/focs/AtseriasGM08}: %
 in the worst case a join query produces
$\Theta(N^{\rho^*})$ tuples, where $\rho^\ast$ is the fractional edge
cover number of the query and $N$ the size of the database.
However, using suitable data structures, much more efficient representations are possible.

In this paper, we consider \emph{$\{\cup,  \times\}$-circuits}
(originally called \emph{d-representations}), which are the most general representation format
studied so far and lead to the most succinct
representations.
In a nutshell, the idea is to decompose the result relation using
union and Cartesian product, where representations of intermediate relations may also be
re-used. The representation can be viewed
as a circuit consisting of $\cup$- and $\times$-gates,
where each gate `computes' a relation and the output gate produces
the query result. 
These $\{\cup,  \times\}$-circuits are
closely related to \DNNF{} circuits from the field of \emph{knowledge
  compilation}, which primarily deals with succinct representations of
Boolean functions. The interaction between database theory and
knowledge compilation has already been quite fruitful and we invite
the reader to consult the recent exposition of Amarilli and Capelli
\cite{DBLP:journals/sigmod/AmarilliC24} for an overview. Our main results
emerge from this connection as we 
use structural
properties of join queries as well as methods from knowledge compilation to
obtain bounds on the representation size.

\begin{example}\label{ex:DB}
  Evaluating $\varphi(x,y,z,w)=R(x,y,z)\land E(x,w)\land
  E(w,y)$ on a database $D$ with
\begin{align*}
         R^D :=& \{a,b,c,d\} \times \{r\}\times \{a,b,c,d\} \; \cup \;
        \{s,t\}\times\{a,b,c,d\}\times \{s,t\} \\
        E^D :=&     \{a,b,c,d\}\times\{r,t\}
        \;\cup\; \{r,s\}\times\{a,b,c,d\}
        \;\cup\; \{t\}\times\{e,f,g,h\}
        \;\cup\; \{e,f,g,h\}\times\{s\}.
\end{align*}
$D$ consists of $N=56$ tuples and the query result $Q(D)$ consists of $48$ tuples, as follows.
\begin{center}
  \addtolength{\tabcolsep}{-0.5em}  \begin{tabular}[h]{cccc|cccc|cccc|cccc||cccc|cccc|cccc|cccc||cccc|cccc|cccc|cccc}
    $x$ & $y$ & $z$ & $w$ &
    $x$ & $y$ & $z$ & $w$ &
    $x$ & $y$ & $z$ & $w$ &
    $x$ & $y$ & $z$ & $w$ &
    $x$ & $y$ & $z$ & $w$ &
    $x$ & $y$ & $z$ & $w$ &
    $x$ & $y$ & $z$ & $w$ &
    $x$ & $y$ & $z$ & $w$ &
    $x$ & $y$ & $z$ & $w$ &
    $x$ & $y$ & $z$ & $w$ &
    $x$ & $y$ & $z$ & $w$ &
    $x$ & $y$ & $z$ & $w$ \\ \hline 
    $a$ & $r$ & $a$ & $r$ &
    $b$ & $r$ & $a$ & $r$ &
    $c$ & $r$ & $a$ & $r$ &
    $d$ & $r$ & $a$ & $r$ &
    $s$ & $a$ & $t$ & $a$ &
    $s$ & $b$ & $t$ & $a$ &
    $s$ & $c$ & $t$ & $a$ &
    $s$ & $d$ & $t$ & $a$ &
    $t$ & $a$ & $s$ & $e$ &
    $t$ & $b$ & $s$ & $e$ &
    $t$ & $c$ & $s$ & $e$ &
    $t$ & $d$ & $s$ & $e$ \\
    $a$ & $r$ & $b$ & $r$ &
    $b$ & $r$ & $b$ & $r$ &
    $c$ & $r$ & $b$ & $r$ &
    $d$ & $r$ & $b$ & $r$ &
    $s$ & $a$ & $t$ & $b$ &
    $s$ & $b$ & $t$ & $b$ &
    $s$ & $c$ & $t$ & $b$ &
    $s$ & $d$ & $t$ & $b$ &
    $t$ & $a$ & $s$ & $f$ &
    $t$ & $b$ & $s$ & $f$ &
    $t$ & $c$ & $s$ & $f$ &
    $t$ & $d$ & $s$ & $f$ \\
    $a$ & $r$ & $c$ & $r$ &
    $b$ & $r$ & $c$ & $r$ &
    $c$ & $r$ & $c$ & $r$ &
    $d$ & $r$ & $c$ & $r$ &
    $s$ & $a$ & $t$ & $c$ &
    $s$ & $b$ & $t$ & $c$ &
    $s$ & $c$ & $t$ & $c$ &
    $s$ & $d$ & $t$ & $c$ &
    $t$ & $a$ & $s$ & $g$ &
    $t$ & $b$ & $s$ & $g$ &
    $t$ & $c$ & $s$ & $g$ &
    $t$ & $d$ & $s$ & $g$ \\
    $a$ & $r$ & $d$ & $r$ &
    $b$ & $r$ & $d$ & $r$ &
    $c$ & $r$ & $d$ & $r$ &
    $d$ & $r$ & $d$ & $r$ &
    $s$ & $a$ & $t$ & $d$ &
    $s$ & $b$ & $t$ & $d$ &
    $s$ & $c$ & $t$ & $d$ &
    $s$ & $d$ & $t$ & $d$ &
    $t$ & $a$ & $s$ & $h$ &
    $t$ & $b$ & $s$ & $h$ &
    $t$ & $c$ & $s$ & $h$ &
    $t$ & $d$ & $s$ & $h$
  \end{tabular}
\end{center}

  $Q(D)$ can be  succinctly represented by the
  following $\{\cup,\times\}$-circuit
  with 9 internal gates and 26
  inputs.\footnote{For simplicity the gates have fan-in larger than 2
  in this example, see Remark~\ref{rem:fan_in}.} 
  The three $\times$-gates correspond to the three blocks from
  the table.
  \begin{center}
    \begin{tikzpicture}
  [wire/.style={thick, -stealth, shorten <= 0.3mm, shorten >= 0.5mm},
  gwire/.style={thick, ->, shorten <= 0.3mm, shorten >= 0.5mm, color=green},
   edge/.style={ -stealth, shorten <= 0.3mm, shorten >= 0.5mm},
   inputgate/.style={inner sep=1pt,minimum size=1mm},
   gate/.style={draw,circle,inner sep=1pt,minimum size=1mm},
   vertex/.style={draw,circle,fill=white,inner sep=2pt,minimum
   size=5.5mm}]

\node[gate]      (g1) at (0,0) {\small $\cup$};
\node[inputgate] (g2) at (1,0) {\scriptsize$y\!\mapsto\! r$};
\node[gate]      (g3) at (2,0) {\small $\cup$};
\node[inputgate] (g4) at (3,0) {\scriptsize$w\!\mapsto\! r$};

\node[inputgate] (g5) at (4,0) {\scriptsize$x\!\mapsto\! s$};
\node[gate]      (g6) at (5,0) {\small $\cup$};
\node[inputgate] (g7) at (6,0) {\scriptsize$z\!\mapsto\! t$};

\node[gate]      (g8) at (7,0) {\small $\cup$};

\node[inputgate] (g9) at (8,0) {\scriptsize$x\!\mapsto\! t$};
\node[inputgate] (g10) at (9,0) {\scriptsize$z\!\mapsto\! s$};
\node[gate]      (g11) at (10,0) {\small $\cup$};

\node[gate]      (x1) at (3,0.75) {\small $\times$};
\node[gate]      (x2) at (6,0.75) {\small $\times$};
\node[gate]      (x3) at (8,0.75) {\small $\times$};
\node[gate]      (u) at (5.5,1.5) {\small $\cup$};

\draw[wire] (g1) -- (x1);
\draw[wire] (g2) -- (x1);
\draw[wire] (g3) -- (x1);
\draw[wire] (g4) -- (x1);
\draw[wire] (g5) -- (x2);
\draw[wire] (g6) -- (x2);
\draw[wire] (g7) -- (x2);
\draw[wire] (g8) -- (x2);
\draw[wire] (g8) -- (x3);
\draw[wire] (g9) -- (x3);
\draw[wire] (g10) -- (x3);
\draw[wire] (g11) -- (x3);

\draw[wire] (x1) -- (u);
\draw[wire] (x2) -- (u);
\draw[wire] (x3) -- (u);

\node[inputgate] (xa) at (-0.5,-0.75) {\scriptsize$x\!\mapsto\! a$};
\node[inputgate] (xb) at (-0.5,-1.25) {\scriptsize$x\!\mapsto\! b$};
\node[inputgate] (xc) at ( 0.5,-0.75) {\scriptsize$x\!\mapsto\! c$};
\node[inputgate] (xd) at ( 0.5,-1.25) {\scriptsize$x\!\mapsto\! d$};

\draw[wire] (xa) -- (g1);
\draw[wire] (xb) .. controls (-0.15,-1) .. (g1);
\draw[wire] (xc) -- (g1);
\draw[wire] (xd) .. controls ( 0.15,-1) .. (g1);

\node[inputgate] (za) at (2-0.5,-0.75) {\scriptsize$z\!\mapsto\! a$};
\node[inputgate] (zb) at (2-0.5,-1.25) {\scriptsize$z\!\mapsto\! b$};
\node[inputgate] (zc) at (2+0.5,-0.75) {\scriptsize$z\!\mapsto\! c$};
\node[inputgate] (zd) at (2+0.5,-1.25) {\scriptsize$z\!\mapsto\! d$};

\draw[wire] (za) -- (g3);
\draw[wire] (zb) .. controls (2-0.15,-1) .. (g3);
\draw[wire] (zc) -- (g3);
\draw[wire] (zd) .. controls (2+0.15,-1) .. (g3);

\node[inputgate] (wa) at (5-0.5,-0.75) {\scriptsize$w\!\mapsto\! a$};
\node[inputgate] (wb) at (5-0.5,-1.25) {\scriptsize$w\!\mapsto\! b$};
\node[inputgate] (wc) at (5+0.5,-0.75) {\scriptsize$w\!\mapsto\! c$};
\node[inputgate] (wd) at (5+0.5,-1.25) {\scriptsize$w\!\mapsto\! d$};

\draw[wire] (wa) -- (g6);
\draw[wire] (wb) .. controls (5-0.15,-1) .. (g6);
\draw[wire] (wc) -- (g6);
\draw[wire] (wd) .. controls (5+0.15,-1) .. (g6);

\node[inputgate] (ya) at (7-0.5,-0.75) {\scriptsize$y\!\mapsto\! a$};
\node[inputgate] (yb) at (7-0.5,-1.25) {\scriptsize$y\!\mapsto\! b$};
\node[inputgate] (yc) at (7+0.5,-0.75) {\scriptsize$y\!\mapsto\! c$};
\node[inputgate] (yd) at (7+0.5,-1.25) {\scriptsize$y\!\mapsto\! d$};

\draw[wire] (ya) -- (g8);
\draw[wire] (yb) .. controls (7-0.15,-1) .. (g8);
\draw[wire] (yc) -- (g8);
\draw[wire] (yd) .. controls (7+0.15,-1) .. (g8);

\node[inputgate] (we) at (10-0.5,-0.75) {\scriptsize$w\!\mapsto\! e$};
\node[inputgate] (wf) at (10-0.5,-1.25) {\scriptsize$w\!\mapsto\! f$};
\node[inputgate] (wg) at (10+0.5,-0.75) {\scriptsize$w\!\mapsto\! g$};
\node[inputgate] (wh) at (10+0.5,-1.25) {\scriptsize$w\!\mapsto\! h$};

\draw[wire] (we) -- (g11);
\draw[wire] (wf) .. controls (10-0.15,-1) .. (g11);
\draw[wire] (wg) -- (g11);
\draw[wire] (wh) .. controls (10+0.15,-1) .. (g11);
  \end{tikzpicture}
  \end{center}
\end{example}

It will be convenient to view join query evaluation as
a homomorphism problem, where the goal is to compute/represent the set
$\Hom(\strucA, \strucB)$ of homomorphisms between a left-hand-side
structure $\strucA$, corresponding to a join query, and a
right-hand-side structure $\strucB$, modeling a relational
database. Moreover, we are interested in the classical data complexity setting,
where $\strucA$ is considered to be fixed and we measure the
complexity in terms of the size $N$ of $\strucB$.
The investigation of this setting was initiated by Olteanu and Z{\'{a}}vodn{\'{y}} \cite{DBLP:journals/tods/OlteanuZ15}, who showed that $\Hom(\strucA, \strucB)$ admits $\{\cup,
\times\}$-circuits of size $O(N^{\fhtw(\strucA)})$, where
$\fhtw(\strucA)$ is the fractional hypertree width of
$\strucA$. Moreover, the
obtained circuits have the additional property of being
\emph{structured} and \emph{deterministic},\footnote{We use
  established terminology from knowledge compilation here, the original name is \emph{deterministic
    d-representations respecting a d-tree}; these are a special case of deterministic, structured $\{\cup,
\times\}$-circuits.}
which in particular allows us to exactly count the number of result
tuples.
It also follows from
\cite{DBLP:journals/jacm/Marx13,DBLP:conf/mfcs/BerkholzS19}, that for
arbitrary small $\delta>0$, there always exist 
general $\{\cup,\times\}$-circuits of size
$O(N^{(1+\delta)\subw(\strucA)})$  that represent all
homomorphisms, where $\subw(\strucA)$ is the \emph{submodular width}
of $\strucA$. These representations can in particular be used for
constant delay enumeration \cite{DBLP:conf/mfcs/BerkholzS19} and
approximate counting \cite{meel2024cfg}. 

\confORfull{Recent progress on lower bounds has been made by Berkholz and
Vinall-Smeeth~\cite{DBLP:conf/icalp/BerkholzV23}, who showed that}{The authors of this paper recently made progress on lower bounds~\cite{DBLP:conf/icalp/BerkholzV23}, by showing that}
if the arity of all relations is bounded by some global constant, then there
is an $\varepsilon>0$ such that for
every $\strucA$ and $N$ there are structures $\strucB$ of size $N$
such that any $\{\cup,\times\}$-circuit representing $\Hom(\strucA,
\strucB)$ has size $N^{\Omega(\tw(\strucA)^\varepsilon)}$, where
$\tw(\strucA)$ is the \emph{treewidth} of $\strucA$. Note that in this
bounded-arity-setting (one may also just think of graphs) we have
$\tw(\strucA)\approx \fhtw(\strucA)\approx\subw(\strucA)$, whereas in
 general one may have $\tw(\strucA)\gg
 \fhtw(\strucA)\gg\subw(\strucA)$.\footnote{ More precisely, they
   disagree only by a factor of $r$ on $r$-ary structures. In
   general, structures can have constant $\subw$ ($\fhtw$) and
   unbounded $\fhtw$ ($\tw$) \cite{DBLP:journals/mst/Marx11}.}

\begin{theorem} \label{thm: tw_struc}
Let $r$ be an integer. Then there is an $\varepsilon>0$ such that for every structure $\strucA$ of arity at most $r$
there exist arbitrarily large structures $\strucB$  such
that any $\{\cup,  \times\}$-circuit for $\Hom(\strucA, \strucB)$ has
size at least $\|B\|^{\varepsilon \tw(\strucA)}$.
\end{theorem}

\noindent Our second main result provides the first lower bound when the arity
is unrestricted. Here we show that the submodular width bound is
almost optimal.

\begin{theorem} \label{thm: sub_lb}
There is an $\varepsilon>0$ such that for every structure $\strucA$
there exist arbitrarily large structures $\strucB$  such
that any $\{\cup,  \times\}$-circuit for $\Hom(\strucA, \strucB)$ has
size at least $\|B\|^{\varepsilon \subw(\strucA)^{1/4}}$.
\end{theorem} 

\subparagraph*{Techniques}
We deploy a similar `recipe' for proving both of our main results. We begin with a known characterisation of what it means for a structure $\strucA$ to have high $x$-width, for $x \in \{\textrm{tree, submodular}\}$. In both cases, the characterisation is of the form `there is a subset of the universe of $\strucA$ which is highly connected'. What highly connected means depends on the notion of width but the idea is that this subset is hard to disconnect, for instance by  removing vertices or edges. Next we have to define some bad instance: that is some structure $\strucB$ such that there is no small $\{\cup,  \times\}$-circuit for $\Hom(\strucA, \strucB)$. We do this by using randomness. This is relatively straightforward for Theorem~\ref{thm: tw_struc}; we first describe this proof on a high level.

For this we use the characterisation of treewidth in terms of a highly connected set from \cite{reed1997tree, robertson1995graph}. This characterisation is also used by Marx to prove a conditional lower bound in the decision setting \cite{DBLP:journals/toc/Marx10}. He\confORfull{}{ uses this characterisation to show} showed that every graph with high treewidth has high `embedding power' and can therefore simulate a $3$-SAT instance\confORfull{}{ in an efficient way}. Thus\confORfull{}{, under an assumption about how quickly $3$-SAT can be solved,} 
Marx obtains conditional lower bounds on the decision problem which are tight up to a logarithmic factor in the exponent. 
\confORfull{}{Moreover, Marx conjectured that this logarithmic factor can be removed\confORfull{.}{; proving this is a major open question.}}
 As we do not need to use embeddings we obtain lower bounds in our setting which are tight up to a constant in the exponent. 
Since the arity is bounded, it turns out that the main step reduces to the special
case where $\strucA$ is a graph.
Here we can set $\strucB$ to be a random graph where each edge is included independently with probability a half. There are two key ideas. First, the set $\Hom(\strucA, \strucB)$ is big; this is essential otherwise we can get a small representation with a trivial circuit which lists every homomorphism. Second, $\strucB$ contains (with high probability) no big bicliques. Intuitively, this makes it hard for a $\{\cup,  \times\}$-circuit to significantly `compress' $\Hom(\strucA, \strucB)$. 

\confORfull{A similar idea is deployed}{We deployed a similar idea} in
\cite{DBLP:conf/icalp/BerkholzV23} to prove a lower bound where the
left-hand structure is a complete graph. The reason we can in general
improve the dependence on the treewidth compared with
\cite{DBLP:conf/icalp/BerkholzV23} is that we use a different method
to lower bound the size of circuits; we deploy a framework
based on tools from \emph{communication complexity} that has been
successfully deployed in knowledge compilation, see
e.g.~\cite{DBLP:conf/uai/BeameL15, DBLP:conf/lics/Capelli17,
  DBLP:journals/mst/BolligF21, DBLP:journals/jair/ColnetM23,
  DBLP:conf/ijcai/Vinall-Smeeth24}, since its introduction
in~\cite{DBLP:conf/ijcai/BovaCMS16}. Using this framework we do not
directly reason about circuits---as is done in
\cite{DBLP:conf/icalp/BerkholzV23}---but instead about sets of
\emph{combinatorial rectangles}  (see Section~\ref{s:rect}). This
allows us to
directly
reason about each left-hand-side instance $\strucA$; in contrast \confORfull{Berkholz and Vinall-Smeeth}{in our previous work we} only directly reason about the case where $\strucA$ is a $k$-clique and then perform a series of reductions which crucially rely on the excluded grid theorem of Robertson and Seymour \cite{DBLP:journals/jct/RobertsonS03a}. We show that if we partition the vertices of a high treewidth graph $\strucA$, such that the highly connected set is evenly split between the two sides of the partition, then there must be a large number of edges crossing the partition. This is the key to us proving our lower bound. 

The proof of Theorem~\ref{thm: sub_lb} follows a similar pattern
but we have to be more careful about defining our `bad instances', since we are now working with higher arity structures. In particular, we have to control the size of each relation in $\strucB$ since if they get too big compared to $|\!\Hom(\strucA, \strucB)|$ it will destroy the lower bound. To do this we exploit the existence of a highly connected set in $\strucA$ and use this explicitly in our construction of $\strucB$. 
Here we use the same characterisation of highly connectedness used in \cite{DBLP:journals/jacm/Marx13} to reason about the decision setting.

\subparagraph*{Further related work} 
Knowledge compilation originally emerged as a sub-field of
artificial intelligence concerned with the efficient representation of
Boolean functions \cite{DBLP:journals/jair/DarwicheM02}.
Within recent years, concepts from this area have been applied more
widely, particularly within database theory (again
see the survey \cite{DBLP:journals/sigmod/AmarilliC24}) and the related
area of constraint satisfaction \cite{DBLP:conf/icalp/BerkholzV23, DBLP:conf/stacs/BerkholzMW24, DBLP:conf/ijcai/KoricheLMT15, DBLP:conf/cp/MateescuD06, DBLP:journals/ijait/AmilhastreFNP14}. An advantage of this approach is that
  if we can represent our query answer as a (particular type of) $\{\cup, \times\}$-circuit, then we can leverage known results to perform 
e.g.
enumeration 
\cite{DBLP:journals/tods/OlteanuZ15, DBLP:conf/icalp/AmarilliBJM17}, 
sampling \cite{DBLP:conf/lpar/SharmaGRM18}, direct access 
\cite{DBLP:conf/icdt/CapelliI24} or (approximate) (weighted) 
counting \cite{DBLP:journals/japll/KimmigBR17, DBLP:conf/stoc/ArenasCJR21a, meel2024cfg} efficiently.\footnote{Note that while some of these results are 
only shown for circuits computing sets of Boolean functions, as is 
observed in the survey \cite{DBLP:journals/sigmod/AmarilliC24}, it is 
straightforward to adapt them to our more general setting.} 
These circuits (and variants) have proved useful for many other aspects of database theory: for instance for MSO queries \cite{DBLP:conf/icalp/AmarilliBJM17}, incremental view maintenance \cite{DBLP:conf/pods/BerkholzKS17, DBLP:conf/pods/Olteanu24}, probabilistic databases \cite{DBLP:conf/icdt/JhaS11} and representing query provenance \cite{DBLP:conf/icdt/OlteanuZ12}.

Recent work \cite{DBLP:conf/icalp/Komarath0R22,DBLP:journals/pacmmod/FanKZ24,DBLP:journals/pacmmod/FanKZ25,DBLP:conf/mfcs/BhargavCC025} studies the algebraic
circuit size for sum-product polynomials over some semiring. While
algebraic circuits are generally incomparable to $\{\cup,
\times\}$-circuits, it turns out that multilinear 
circuits over the Boolean semiring considered in \cite{DBLP:journals/pacmmod/FanKZ24,DBLP:journals/pacmmod/FanKZ25}  can be translated to $\{\cup,\times\}$-circuits with only a polynomial blow-up, but not vice versa. In this sense,
our lower bounds hold in a strictly more general model and also imply
lower bounds for these algebraic circuits, see Appendix~\ref{a:algebraic}.

As observed by Feder and Vardi \cite{DBLP:conf/stoc/FederV93} the problem of deciding if there is a homomorphism between two relational structures $\strucA$ and $\strucB$ is a generalisation of the constraint satisfaction problem. Given the generality of this problem it is unsurprisingly NP-complete. Therefore, research has focused on finding tractable restrictions, where we demand $\strucA \in \classA$ and $\strucB \in \classB$ for classes of structures $\classA, \classB$; we denote this problem by $\CSP(\classA, \classB)$. If we set $\classA$ to be the class of all structures $\classall$ then we have a complete (conditional) understanding of when we have polynomial time tractability for $\CSP(\classall, \classB)$ \cite{DBLP:conf/focs/Bulatov17, DBLP:journals/jacm/Zhuk20}. 

However, for \emph{left-hand-side restrictions,} where we set $\classB = \classall$ and ask for what $\classA$ we have polynomial time solvability, no dichotomy is known. 
Still, a lot of work has been done on such restrictions.  
If $\classA$ has \emph{bounded arity} we know that $\CSP(\classA, \classall)$ is in polynomial time iff there is a constant bounding the treewidth of the \emph{homomorphic core} of every structure in $\classA$ \cite{DBLP:journals/jacm/Grohe07, DBLP:journals/toc/Marx10}. Moreover, for the counting variant the criteria is that $\classA$ has bounded treewidth \cite{DBLP:journals/tcs/DalmauJ04}. In the general case, we do have a conditional dichotomy in the parametrised setting: there the correct criteria is that $\classA$ has bounded submodular width \cite{DBLP:journals/jacm/Marx13}. Our main result imply comparable dichotomy results for the problem of representing $\Hom(\strucA, \strucB)$ by a $\{\cup, \times\}$-circuit; see Section~\ref{ss:rep} for more details.

        \section{Preliminaries}\label{sec:preliminaries}

\paragraph{Functions and Sets}
For $n \in \mathbb{N}$, we write $[n]:= \{1, \dots, n\}$. Given a set $S$
we write $\mathcal{P}(S)$ to denote the power set of $S$.\confORfull{}{ Moreover, $(X,Y)$---a pair of disjoint subsets of $S$---is a \emph{partition} of $S$ if $X \cup Y = S$.} Whenever writing $\overline{a}$ to denote
a $k$-tuple, we write $a_i$ to denote its $i$-th component;
i.e., $\overline{a}=(a_1,\dots, a_k)$. Moreover, we write $\tilde{a}$
to denote the set $\{a_1, \dots, a_k\}$. For a function $f \colon X
\to Y$ and $X' \subset X$ we write $f\restriction_{X'}$
to denote the restriction of $f$ to $X'$, i.e. $f\restriction_{X'} \colon X' \to Y$ is defined by $x \mapsto f(x)$ for all $x \in X'$. We extend this operation to a set of such functions $F$ by setting
$F\restriction_{X'} \, := \{f\restriction_{X'} \, \mid \, f \in F\}$. Given disjoint sets $X, X'$ and two functions $f \colon X \to Y, g\colon X' \to Y'$ we define $f \times g \colon X \cup X' \to Y \cup Y'$ by:
\[
(f \times g) (x) = \begin{cases}
f(x) & x \in X \\
g(x) & x \in X'.
\end{cases}
\]
We will often combine two sets of functions pairwise in this way. Formally, for a set $S$ of functions with domain $X$ and a set $T$ of functions with domain $X'$ we define $S \times T := \{f \times g \, \mid \, f \in S,\, g \in T\}$. We also use $\times$ to denote Cartesian product in the normal sense; we can overcome this seeming mismatch by viewing a function $f \colon X \to D$ as a tuple with coordinates in bijection with $X$ such that the $x$-th coordinate is $f(x)$.

\subparagraph*{Graphs and Hypergraphs}
A \emph{hypergraph} $\graphH = (\vertexH, \edgeH)$ is a pair, consisting of a finite set of vertices $\vertexH$ and a finite set of edges $\edgeH \subseteq 2^{\vertexH}$. If $|e| = 2$, for every $e \in \edgeH$, we say that $\graphH$ is a \emph{graph}. Via this definition we are implicitly assuming that all graphs are undirected and \emph{simple}, i.e., that they do not contain self-loops. We write $\graphK_k$ for
the complete graph on $k$ vertices, which we also call the \emph{$k$-clique}. 
Given a hypergraph $\graphH = (\vertexH, \edgeH)$ and a subset $X \subseteq \vertexH$,
an \emph{edge cover} of $X$ is a set $S \subseteq \edgeH$ such that for all $ v\in X$ there is some $e \in S$ with $v \in e$. A \emph{fractional edge cover} of $X$ is the natural fractional relaxation, i.e.
 a function $s \colon \edgeH \to [0,1]$, such that for every $v \in X$, we have that
$\sum_{e \ni v} s(e) \ge 1$. We write $\weight(s):= \sum_{e \in
  \edgeH} s(e)$.
  Then $\fec(X)$---the \emph{fractional edge cover
  number of $X$}---is defined to be the minimum of $\weight(s)$ over
all fractional edge covers of $X$. %

We extend functions $f \colon \vertexH \to \mathbb{R}^{\ge 0}$ to subsets of
$X \subseteq \vertexH$ in the natural way, i.e.\! by setting $f(X) =
\sum_{v \in X} f(v)$. A \emph{fractional independent set} of $\graphH$
is a map $\mu \colon \vertexH \to [0,1]$, such that for every $e  \in
\edgeH$, $\sum_{v \in e} \mu(v) \le 1$; or equivalently $\mu(e) \le 1$
for every $e \in \edgeH$. Note that by LP duality the weight
$\mu(\vertexH)$ of minimal fractional independent set $\mu$ agrees with
the weight of a maximal fractional edge cover. 
Now let $\mathbf{S}(\graphH)$ be the set of all functions $f \colon \mathcal{P}(\vertexH) \to \mathbb{R}^{\ge 0}$ with $f(\emptyset) = 0$ that satisfy the following properties:
\begin{enumerate}
\item \textbf{Monotonicity:} $f(S) \le f(T)$ for all $S \subseteq T \subseteq \vertexH$,
\item \textbf{Edge dominatedness:} $f(e) \le 1$ for every $e \in \edgeH$, and 
\item \textbf{Submodularity:} $f(S) + f(T) \ge f(S \cap T) + f(S \cup T)$ for every $S, T \subseteq \vertexH$. 
\end{enumerate}
Observe that for every fractional independent set $\mu$ of $\graphH$ we have
$\mu\in \mathbf{S}(\graphH)$.

\subparagraph*{Width Measures}
A \emph{tree decomposition} of a hypergraph $\graphH$ is a pair
$(\treeT,\beta)$ where $\treeT$ is a tree and $\beta\colon \vertexT \to \mathcal{P}(\vertexH)$ associates to every node $t\in \vertexT$ \emph{a bag} $\beta(t)$ such
that:
\begin{enumerate}
\item for every $v\in \vertexH$ the set $\{t \in \vertexT \mid v\in
  \beta(t)\}$ is a non-empty connected subset of $\vertexT$ and
\item for every $e \in \edgeH$ there is some $t\in \vertexT$ such that $ e \subseteq
  \beta(t)$.
\end{enumerate}
We write $\classT(\graphH)$ for the set of all tree decompositions of $\graphH$. The width of a tree decomposition is $\max_{t\in \vertexT}|\beta(t)|-1$ and the 
\emph{treewidth} of $\graphH$ is the minimum width of all
 $(\treeT,\beta) \in \classT(\graphH)$. \confORfull{}{Note that the -1 in the definition of treewidth is there to ensure that trees have treewidth one.}

\confORfull{}{From an algorithmic point of view, the idea of treewidth is that we can perform various bottom-up computations on the tree decomposition. Since we only consider one bag at a time, if the bags are small this can allow for efficient algorithms on small treewidth hypergraphs for problems which are in general intractable. For this, and other reasons, treewidth has become an incredibly influential concept and has inspired many more sophisticated notions of width, see~\cite{DBLP:journals/cj/HlinenyOSG08} for a survey. } 

\confORfull{Treewidth is an incredibly influential concept and has inspired many more sophisticated notions of width, see~\cite{DBLP:journals/cj/HlinenyOSG08} for a survey.}{} Two such measures are relevant in this paper: fractional hypertree width and submodular width. The \emph{fractional hypertree width} of a tree decomposition $(\treeT,\beta)$ is defined to be the maximum value of $\fec(\beta(t))$ over all $t \in \vertexT$. The fractional hypertree width of a hypergraph $\graphH$ is the minimum fractional hypertree width over all tree decomposition of $\graphH$, in symbols
\[
\fhtw(\graphH) := \min_{(\treeT,\beta) \in \classT(\graphH)} \, \max_{t \in \vertexT} \fec(\beta(t)).
\]
So if $\graphH$ has low fractional hypertree width then there is a
tree decomposition where every bag has small fractional edge cover
number.

The definition of submodular width is slightly more involved: if $\graphH$ has low
submodular width then for \emph{every} $ f \in \mathbf{S}(\graphH)$
 there is a
tree decomposition where every bag evaluates to a small number
under $f$.
To be
precise, the submodular width of $\graphH$ is defined to be
\[
\subw(\graphH) := \sup_{f \in \mathbf{S}(\graphH)} \min_{(\treeT,\beta) \in \classT(\graphH)} \, \max_{t \in \vertexT} f(\beta(t)).
\]
From the definitions, one can see that $\tw(\graphH) + 1 \ge \fhtw(\graphH)$. Moreover, Marx proved that $\fhtw(\graphH) \ge \subw(\graphH)$~\cite[Proposition 3.7.]{DBLP:journals/jacm/Marx13}.

\subparagraph*{Structures and Homomorphisms} \label{ss:struc_hom}

A \emph{(relational) signature} $\sigma$ is a set of relation symbols $R$, each  equipped with an arity $r =r(R)$. We assume all signatures are finite. A \emph{(finite, relational) $\sigma$-structure} $\strucA$
consists of a finite set $A$ (the \emph{universe} of $\strucA$) and  relations
$R^{\strucA}\subseteq A^r$ for every $r$-ary relation
symbol $R\in \sigma$. The \emph{size} of $\strucA$ is defined to be $\|\strucA\|:=
\sum_{R \in \sigma} |R^{\strucA}|$ and the \emph{arity} of $\strucA$ is the maximum arity of all relations in $\sigma$.

The \emph{Gaifman graph} of 
$\strucA$ is the graph with vertex set $A$ and edges $\{u,v\}$ for any
distinct $u,v$ that occur together in a tuple of a relation in
$\strucA$. We say $\strucA$ is \emph{connected} if its Gaifman graph
is connected and we will henceforth assume that all structures are
connected. \confORfull{}{All of our results can be easily extended to
  disconnected structures.} The \emph{hypergraph} of a
$\sigma$-structure $\strucA$ is a hypergraph with vertex set $A$ such
that $e \subseteq A$ is an edge iff there is a tuple $\overline{e} \in
R^{\strucA}$ for some $R \in \sigma$, such that $\tilde{e} = e$.

For $x \in \{\textrm{tree, fractional hypertree, submodular}\}$, the \emph{$x$-width} of a structure is the $x$-width of its
hypergraph.\confORfull{}{ Note that the tree-width of a structure is equal to the tree-width of its Gaifman graph. This is easy to see, using the fact that if $(\treeT, \beta)$ is a tree decomposition of a graph $\graphG$ and $X$ induces a clique in $\graphG$, then $X$ is contained in some bag of the tree decomposition, i.e.\! $X \subseteq \beta(t)$ for some $t \in \vertexT$.} For a class of structures $\classA$, we write $\tw(\classA) := \sup_{\strucA \in \classA} \tw(\strucA)$, which may be infinite. If $\tw(\classA)$ is finite we say that $\classA$ has \emph{bounded treewidth}, otherwise we say it has unbounded treewidth. We define the analogous notions for fractional hypertree and submodular width. It is known that there are classes of structures with bounded fractional hypertree width but unbounded treewidth \cite{DBLP:journals/talg/GroheM14} and a class of structures with bounded submodular width but unbounded fractional hypertree width \cite{DBLP:journals/mst/Marx11}. Finally, if there is some integer $r$ such that every structure in $\classA$ has arity at most $r$, we say that $\classA$ has \emph{bounded arity}. In this case, $\classA$ has bounded treewidth iff $\classA$ has bounded fractional hypertree width iff $\classA$ has bounded submodular width. 

\confORfull{Finally, a}{Finally, given two $\sigma$-structures
  $\strucA, \strucB$ a} \emph{homomorphism from $\strucA$ to
  $\strucB$} is a function $h \colon A \to B$ that preserves all
relations, i.e., for every $r$-ary $R \in \sigma$ and $\bar{a} =(a_1,
\dots, a_r) \in R^{\strucA}$ it holds that $h(\bar{a}) := (h(a_1),
\dots, h(a_r)) \in R^{\strucB}$. We write $\Hom(\strucA, \strucB)$ for
the set of all homomorphisms from $\strucA$ to $\strucB$.

        \section{Homomorphisms and \texorpdfstring{$
                        \{\cup, \times\}$-}{Union-Product }Circuits}
        \label{s:circuit}
In this section we define $\{\cup, \times\}$-circuits and explain how they can be used to represent sets of homomorphisms. We roughly follow the presentation of Amarilli and Capelli \cite{DBLP:journals/sigmod/AmarilliC24}. 
Let $(X,D)$ be a pair of finite sets. Then a \emph{$\{\cup,
  \times\}$-labelled {\upshape DAG} over $(X,D)$}
 is a vertex-labelled directed acyclic graph with a unique sink such that every internal (i.e.\! non-source) node has fan-in two and is labelled by either $\times$ or $\cup$ and every source is labelled with an expression of the form $x \mapsto d$, where $x \in X$ and $d \in D$. We call the set $X$ the \emph{variables} and the set $D$ the \emph{domain}. We call nodes labelled by $\circ \in \{\cup, \times\}$ $\circ$-gates and sources input-gates or just inputs. 
 
 For $\C$
 a $\{\cup, \times\}$-labelled DAG over $(X,D)$ and $g$ a node of
 $\C$, we say that a variable $x \in X$ \emph{occurs below} $g$ if
 there is an input node $g'$ labelled with $x
 \mapsto d$ for some $d \in D$, such that there is a directed path from $g'$ to $g$ in
 $\C$. We write $\var(g)$ for the set of variables occurring below $g$
 and $\var(\C)$ for the set of all variables occurring in the DAG. With these notions in hand we can now define the main representation format studied in this paper.

\begin{definition}[$\{\cup, \times\}$-circuit]\label{def:circuit}
A $\{\cup, \times\}$-circuit over $(X,D)$ is a $\{\cup,
\times\}$-labelled \textup{DAG} $\C$ over $(X,D)$ with $\var(\C) = X$ such that
\begin{enumerate}
\item for every $\times$-gate $g$ in $\C$ with children $g_1, g_2$ we have $\var(g_1) \cap \var(g_2) = \emptyset$, and
\item for every $\cup$-gate $g$ in $\C$ with children $g_1, g_2$ we have $\var(g_1) = \var(g_2)$.
\end{enumerate} 
\end{definition}

\begin{figure}  
  \centering
  \begin{tikzpicture}
  [wire/.style={thick, ->, shorten <= 0.3mm, shorten >= 0.5mm},
  gwire/.style={thick, ->, shorten <= 0.3mm, shorten >= 0.5mm, color=green},
   edge/.style={ -stealth, shorten <= 0.3mm, shorten >= 0.5mm},
   inputgate/.style={inner sep=1pt,minimum size=1mm},
   gate/.style={draw,circle,inner sep=1pt,minimum size=1mm},
   vertex/.style={draw,circle,fill=white,inner sep=2pt,minimum
   size=5.5mm}]

     \node at (0.5,-2) {$\mathcal{G}$};
    \node at (3.5,-2) {$\mathcal{H}$};
    
    \node at (9,-2) {A $\{\cup, \times\}$-circuit computing $\Hom(\mathcal{G},\mathcal{H})$};
 
        \def\y{0.5}
 
    \node[vertex,] (x1) at (0,0-\y) {$x$};
    \node[vertex,] (x2) at (1,.577-\y) {$y$};
    \node[vertex,] (x3) at (1,-.577-\y) {$z$};
    \draw[edge] (x1) -- (x2);
    \draw[edge] (x1) -- (x3);
    \draw[edge] (x2) -- (x3);
    
        \node[vertex,] (a1) at (2.5,0+.6-\y) {\small $a_1$};
    \node[vertex,] (a2) at (2.5,0-\y) {\small $a_2$};

    \node[vertex,] (b1) at (3.35,1.4-\y) {\small $b_1$};
    \node[vertex,] (b2) at (4.15,1.4-\y) {\small $b_2$};

    \node[vertex] (d1) at (3.35,-0.8-\y) {\small $d_1$};
    \node[vertex,] (d2) at (4.15,-0.8-\y) {\small $d_2$};

    \node[vertex,] (c) at (2.5+2.5,0.3-\y) {$c$};

    \draw[edge] (a1) -- (b1);
    \draw[edge] (a2) -- (b1);

    \draw[edge] (a1) -- (b2);
    \draw[edge] (a2) -- (b2);

    \draw[edge] (a1) -- (d1);
    \draw[edge] (a2) -- (d1);

    \draw[edge] (a1) -- (d2);
    \draw[edge] (a2) -- (d2);

    \draw[edge] (b1) -- (c);
    \draw[edge] (b2) -- (c);

    \draw[edge] (c) -- (d1);
    \draw[edge] (c) -- (d2);

    \draw[edge] (a1) -- (c);
    \draw[edge] (a2) -- (c);

    \begin{scope}[xshift=9cm,yshift=-1.5cm]
      \node[inputgate] (xa1) at (-.5,0) {\scriptsize$x\!\mapsto\! a_1$};
      \node[inputgate] (xa2) at (+.5,0) {\scriptsize$x\!\mapsto\! a_2$};

      \node[inputgate] (yb1) at (-2.5,0) {\scriptsize$y\!\mapsto\! b_1$};
      \node[inputgate] (yb2) at (-2.5+1,0) {\scriptsize$y\!\mapsto\! b_2$};

      \node[inputgate] (zd1) at (2.5-1,0) {\scriptsize$z\!\mapsto\! d_1$};
      \node[inputgate] (zd2) at (2.5,0) {\scriptsize$z\!\mapsto\! d_2$};

      \node[inputgate] (zc) at (-3.5,0) {\scriptsize$z\!\mapsto\! c$};
      \node[inputgate] (yc) at (3.5,0) {\scriptsize$y\!\mapsto\! c$};

      \node[gate] (g1) at (-1.5,0.75) {\small $\cup$};
      \node[gate] (g2) at (0,0.75) {\small $\cup$};
      \node[gate] (g3) at (1.5,0.75) {\small $\cup$};
      
                \node[gate](e1) at (-0.75,1.5) {\small $\times$};
                \node[gate](e2) at (0.75,1.5) {\small $\times$};

      \node[gate] (g4) at (-1,2.25) {\small $\times$};
      \node[gate] (g5) at (1,2.25) {\small $\times$};
      
      \node[gate] (g6) at (0,3) {\small $\cup$};

      \draw[wire] (yb1) -- (g1);
      \draw[wire] (yb2) -- (g1);
      
      \draw[wire] (xa1) -- (g2);
      \draw[wire] (xa2) -- (g2);
      
      \draw[wire] (zd1) -- (g3);
      \draw[wire] (zd2) -- (g3);

      \draw[wire] (zc) -- (g4);
      \draw[wire] (yc) -- (g5);

      \draw[wire] (g1) -- (e1);
      \draw[wire] (g2) -- (e1);
      
      \draw[wire] (e1) -- (g4);
      \draw[wire] (e2) -- (g5);

      \draw[wire] (g3) -- (e2);
      \draw[wire] (g2) -- (e2);

      \draw[wire] (g4) -- (g6);
      \draw[wire] (g5) -- (g6);

    \end{scope}
\end{tikzpicture}
 \caption{A deterministic $\{\cup, \times\}$-circuit computing the set of all homomorphisms from
  $\graphG$ to $\graphH$.}
\label{fig:representation}
\end{figure}

\noindent 
Such circuits compute a set of mappings\confORfull{}{ from $X \to D$} bottom up by taking unions at $\cup$-gates and Cartesian products at {$\times$}-gates. Conditions (1) and (2), called \emph{decomposability} \cite{DBLP:journals/jacm/Darwiche01} and \emph{smoothness} respectively in the knowledge compilation literature, ensure this process is well-defined. 
Formally, we define the semantics of  $\{\cup, \times\}$-circuit $\C$ by associating with each gate $g$ the set of functions computed by the subcircuit of $\C$ rooted at $g$. We define this as follows.
\[
S(g) := \begin{cases}
\{x \mapsto d\} & g \textit{ is an input gate labelled with } x \mapsto d. \\
S(g_1) \times S(g_2) & g \textit{ is a $\times$-gate with children } g_1, g_2. \\
S(g_1) \cup S(g_2) & g \textit{ is a $\cup$-gate with children } g_1, g_2.
\end{cases}
\]
Note that in the $\times$-gate case, $S(g_1) \times S(g_2) = \{f \times g \, \mid \, f \in S(g_1), \, g \in S(g_2)\}$ is well-defined by (1). Moreover, by (2) we obtain that every function in $S(g)$ has domain $\var(g)$. We write $S(\C):= S(s)$ where $s$ is the unique sink of $\C$. This leads to the following.

\begin{definition}
Let $\strucA, \strucB$ be structures and $\C$ be a $\{\cup, \times\}$-circuit over $(A,B)$. Then we say that $\C$ \emph{computes} $\Hom(\strucA, \strucB)$ if $S(\C) = \Hom(\strucA, \strucB)$. 
\end{definition}

\noindent As in the above definition, 
if  a $\{\cup, \times\}$-circuit computes $\Hom(\strucA, \strucB)$ then we always assume that $X = A$ and $D = B$. We define the size of $\C$, a $\{\cup, \times\}$-circuit, to be the number of vertices in the underlying graph and denote this by $|\C|$. 

\confORfull{}{Before explaining why such representations are useful we introduce two further properties: determinism and structuredness.}
A $\{\cup, \times\}$-circuit $\C$ is \emph{deterministic} if the subcircuits rooted at the children of every $\cup$-gate compute disjoint sets of functions \cite{DBLP:journals/jancl/Darwiche01}. Formally, we require that for every $\cup$-gate $g \in \C$, with children $g_1, g_2$, we have that $S(g_{1}) \cap S(g_{2}) = \emptyset$. Given a deterministic $\{\cup, \times\}$-circuit $\C$ one can compute $|S(\C)|$ in linear time by going bottom up in the circuit, adding at every $\cup$-gate and multiplying at every $\times$-gate.

\confORfull{}{To define structuredness we first need the notion of a \emph{v-tree} over a set of variables $X$. This is is a full, rooted, binary tree $\treeT$ whose leaves are in 1-1 correspondence with the elements of $X$. For $t \in \vertexT$ we write $\var(t)$ to denote the set of variables occurring as leaves in the subtree rooted at $t$. A $\{\cup, \times\}$-circuit $\C$ \emph{respects} a v-tree $T$, if for every $\times$-gate $g \in \C$ with children $g_1, g_2$, there is some $t \in \vertexT$ with children $t_1, t_2$ such that $\var(g_1) \subseteq \var(t_{1})$ and $\var(g_{2}) \subseteq \var(t_{2})$. A $\{\cup, \times\}$-circuit is \emph{structured} if it respects some v-tree \cite{DBLP:conf/aaai/PipatsrisawatD08}. The idea is that decomposability always ensures every $\times$-gate partitions the set of variables occurring below it and structuredness ensures that this partition is done in a uniform way across the circuit.}

\begin{remark} \label{rem:fan_in}
Our stipulation that $\{\cup,\times\}$-circuits are fan-in two is mainly a matter of taste since these circuits can simulate arbitrary fan-in circuits with only a quadratic size blow up. We make this choice since it is more convenient for proving lower bounds (see Section~\ref{s:rect}) but Theorems~\ref{thm: tw_struc} and \ref{thm: sub_lb} also hold---via this quadratic simulation---for unbounded fan-in $\{\cup, \times\}$-circuits. Similarly smoothness\confORfull{}{---condition (2) in our definition of $\{\cup, \times\}$-circuits---}is only enforced for convenience; one can get a sensible representation format without this condition, see e.g. \cite{DBLP:journals/sigmod/AmarilliC24}. Moreover, one can turn an `unsmooth' $\{\cup, \times\}$-circuit into an equivalent smooth one in quadratic time \cite{DBLP:journals/jair/DarwicheM02}. 
\end{remark}

\subsection{Upper Bounds On Representation Size} \label{ss:rep}

We now formally define the problem studied in this paper. Given a class of structures $\classA$ the problem $\{\cup,  \times\}$-$\CSP(\classA, \classall)$ is, given $\strucA \in \classA$ and a second structure $\strucB$ as input, to produce a $\{\cup, \times\}$-circuit computing $\Hom(\strucA, \strucB)$. In the rest of this section we lay out the state of the art for this problem starting with the following result by Olteanu and Z{\'{a}}vodn{\'{y}}.

\begin{theorem}[{\cite{DBLP:journals/tods/OlteanuZ15}}]
If $\classA$ has bounded fractional hypertree width then $\{\cup,  \times\}$-$\CSP(\classA, \classall)$ is polynomial time solvable. More precisely, given structures $\strucA, \strucB$ and $\treeT \in \classT(\strucA)$ with fractional hypertree width $w$, we can construct a\confORfull{}{ structured and} deterministic $\{\cup,  \times\}$-circuit computing $\Hom(\strucA, \strucB)$ of size $O(\|\strucA\|^2 \|\strucB\|^{w} )$ in time $O(\poly(\|\strucA\|) \|\strucB\|^{w} \log(\|B\|))$.
\end{theorem}

\noindent Since the above algorithm outputs a deterministic $\{\cup,  \times\}$-circuit this allows us to compute $|\!\Hom(\strucA, \strucB)|$ in polynomial time given that $\strucA$ has bounded fractional hypertree 
width. Moreover, the output circuit also admits efficient algorithms for enumeration \cite{DBLP:journals/tods/OlteanuZ15, DBLP:conf/icalp/AmarilliBJM17}, sampling \cite{DBLP:conf/lpar/SharmaGRM18} and direct access \cite{DBLP:conf/icdt/CapelliI24}.
In the bounded arity setting classes of bounded fractional hypertree width (= classes of bounded treewidth) are the only classes where $\{\cup,  \times\}$-$\CSP(\classA, \classall)$ is polynomial time solvable. This follows from a lower bound \confORfull{by Berkholz and Vinall-Smeeth.}{proved in our previous work.}

\begin{theorem}[{\cite{DBLP:conf/icalp/BerkholzV23}}] \label{thm:weak_tw_lb}
Let $r$ be an integer. Then there is an $\varepsilon>0$ such that for every structure $\strucA$ of arity at most $r$
there exist arbitrarily large structures $\strucB$  such
that any $\{\cup,  \times\}$-circuit for $\Hom(\strucA, \strucB)$ has
size at least $\|B\|^{\tw(\strucA)^{\varepsilon}}$.
\end{theorem}

\noindent Our Theorem~\ref{thm: tw_struc} improves the exponent to $\Omega(\tw(\strucA))$. This shows that---on classes of bounded arity structures---the algorithm from Olteanu and Z{\'{a}}vodn{\'{y}} is optimal, in that we cannot improve the exponent by more than a constant factor.

When we do not impose a bound on the arity the situation is much less understood; we do not even have a good understanding of when $\CSP(\classA, \classall)$ is polynomial time solvable. We have a better grasp if we allow FPT running times, i.e.\! one that may depend arbitrarily on $\|\strucA\|$ but only polynomially on $\|\strucB\|$. In this setting $\{\cup,  \times\}$-$\CSP(\classA, \classall)$ is tractable for some classes of structures with unbounded fractional hypertree width.

 \begin{theorem}[{\cite{DBLP:journals/jacm/Marx13,DBLP:conf/mfcs/BerkholzS19}}] \label{thm: subw}
   If $\classA$ has bounded submodular width then $\{\cup,
   \times\}$-$\CSP(\classA, \classall)$ is \textup{FPT} time solvable. More
   precisely, for all structures $\mathcal{A}$ and $\mathcal{B}$ we can construct a $\{\cup, \times\}$-circuit computing $\Hom(\strucA, \strucB)$ of size $O(g(\|\mathcal{A}\|) \cdot
   \|\mathcal{B}\|^{(1+\delta)\subw(\strucA)})$ in time $O(h(\|\mathcal{A}\|) 
 \cdot  \|\mathcal{B}\|^{(1+\delta)\subw(\strucA)})$, where $g, h$ are computable
 functions and  $\delta >0$ can be arbitrary small. 
   \end{theorem}
   
\noindent Note that \cite{DBLP:journals/jacm/Marx13,DBLP:conf/mfcs/BerkholzS19} do not mention $\{\cup,  \times\}$-circuits; they are instead concerned with enumeration and decision algorithms respectively. But it is easy to adapt their algorithm to output the required circuit, see Appendix~A in the full version of \cite{berkholz2022dichotomy} for more details. 

Our Theorem~\ref{thm: sub_lb} implies that if $\classA$ has unbounded submodular width then there is no FPT time algorithm for $\{\cup,  \times\}$-$\CSP(\classA, \classall)$. However, there is still a $w^{1/4}$ vs $w$ gap in the exponents\confORfull{.}{; we leave the task of closing this gap to future work.} Another thing to note is that since the $\{\cup,  \times\}$-circuits from Theorem~\ref{thm: subw} are \confORfull{not deterministic}{neither structured nor deterministic}, the resulting representations of $\Hom(\strucA, \strucB)$ are not as useful algorithmically. Despite this, they can still be used to achieve constant delay enumeration with FPT preprocessing \cite{DBLP:conf/mfcs/BerkholzS19} and approximate counting in FPT time~\cite{meel2024cfg}.

        \section{Lower Bounds for \texorpdfstring{$\{\cup, \times\}$-}{Union-Product }Circuits via Rectangle Covers} \label{s:rect}

In this section, we introduce a by now standard framework---pioneered in \cite{DBLP:conf/ijcai/BovaCMS16}---for proving lower bounds on $\{\cup,\times\}$-circuits  and their subclasses: analysing rectangle covers. This enables us to reduce reasoning about complicated objects ($\{\cup,\times\}$-circuits) to  simple ones (combinatorial rectangles). For this we need the following definition.

\begin{definition} \label{def:rect}
Let $X$ be a finite set, $(Y,Z)$ a partition of $X$ and $\mathcal{R}$ be a set of functions with domain $X$. Then $\mathcal{R}$ is an \emph{$(X,Y)$-rectangle} iff $\mathcal{R} = \mathcal{R}\restriction_Y \times \, \mathcal{R}\restriction_Z$. 
\end{definition}

\begin{example} \label{ex:defrec}
Let $X= \{y,z\}$, $Y =\{y\}$ and $Z = \{z\}$. Further let $\mathcal{R}_1 = \{f, g \}$, where $f(y) = f(z) = 0$ and $g(y) = g(z) = 1$. Then $\mathcal{R}_1$ is \emph{not} a $(Y,Z)$-rectangle. To see this note that $ \mathcal{R}\restriction_Y \times \mathcal{R}\restriction_Z$ contains the map $\{y \mapsto 0, z \mapsto 1\}$. On the other hand, the set of functions $\mathcal{R}_2 = \{f \, \mid \, f \colon  \{y,z\} \to \{0,1\}\}$ is a $(Y,Z)$-rectangle. 
\end{example}
\noindent Note that trivially every set of functions with domain $X$ is a $(X, \emptyset)$-rectangle. To get a non-trivial notion we introduce the notion of a \emph{balanced} partition; that is a partition $(Y,Z)$ of $X$ such that $\min(|Y|, |Z|) \ge |X|/3$. We then say a set $\mathcal{R}$ is a \emph{balanced} rectangle, if it is a $(Y,Z)$-rectangle for some balanced partition $(Y,Z)$. We say that a set of balanced rectangles $\{\mathcal{R}_i\}_{i \in \mathcal{I}}$ is a \emph{balanced rectangle cover} for $\Hom(\strucA, \strucB)$ if $\bigcup_{i \in \mathcal{I}} \mathcal{R}_i = \Hom(\strucA, \strucB)$. The following lemma shows that if there is a small $\{\cup,\times\}$-circuit computing $\Hom(\strucA, \strucB)$, then there is a small balanced rectangle cover for $\Hom(\strucA, \strucB)$. 

\begin{lemma}[{\cite{DBLP:conf/ijcai/BovaCMS16, DBLP:conf/stacs/BerkholzMW24}}] \label{lem: basic_rect}
Let $\strucA, \strucB$ be structures and let $\C$ be a $\{\cup,\times\}$-circuit of size $s$ computing $\Hom(\strucA, \strucB)$. Then there is a balanced rectangle cover of $\Hom(\strucA, \strucB)$ of size $s$. 
\end{lemma}

\noindent Therefore, to prove lower bounds on $\{\cup,\times\}$-circuits it suffices to prove lower bounds on the size of balanced rectangle covers. This result was proven for Boolean functions in \cite[Theorem 6]{DBLP:conf/ijcai/BovaCMS16} and extended to general sets of functions in \cite[Lemma 26]{berkholz2023characterization}. 

In fact, we need a slight generalisation of Lemma~\ref{lem: basic_rect} that allows us to balance rectangles with respect to some weighting function $f \colon \mathcal{P}(A) \to \mathbb{R}^{\ge 0}$ defined on subsets of the universe of $\strucA$. To be precise, we say that a partition $(X,Y)$ of $A$ is $f$-balanced if $\min(f(X), f(Y)) \ge f(A)/3$. We define the notions of an $f$-balanced rectangle and an $f$-balanced cover in the obvious way. To ensure the proof goes through we need an $f$ which is sufficiently small on singletons and which is \emph{sub additive on disjoint subsets}. That is on disjoint sets $S,T \subseteq A$, it holds that $f(S \cup T) \le f(S) + f(T)$. \confORfull{See Appendix~\ref{a:weighted} for a proof.}{}

\begin{lemma} \label{lem:weightedrect}
Let $\strucA, \strucB$ be structures, $\C$ be a $\{\cup,\times\}$-circuit  computing $\Hom(\strucA, \strucB)$ and $f \colon \mathcal{P}(A) \to \mathbb{R}^{\ge 0}$ be a function which is sub additive on disjoint subsets and satisfying $f(\{a\}) \le 2f(A)/3$ for all $a \in A$. Then there is a $f$-balanced rectangle cover of $\Hom(\strucA, \strucB)$ of size $|\C|$. 
\end{lemma}%

\confORfull{}{%
\begin{proof}[Proof (sketch)]
Every gate $g$ in a $\{\cup,\times\}$-circuit corresponds to a $(\var(g), A \setminus \var(g))$ rectangle, so if we delete the gate from the circuit we remove the corresponding rectangle from the output. Therefore, by iteratively deleting gates we get a rectangle cover and since the number of iterations is at most the number of gates in the original circuit the size bound on the cover follows;  see the proof of \cite[Lemma 26]{berkholz2023characterization}. It remains to show that at each stage we can find a gate to eliminate with an appropriate partition; i.e.\! one that is $f$-balanced. In fact, it is enough to find a gate $\hat{g}$ such that $f(A)/3 \le f(\var(\hat{g})) \le 2f(A)/3$. To see this note that since $f$ is subadditive on disjoint subsets, $f(A \setminus \var(\hat{g})) \ge f(A) - f(\var(\hat{g})) \ge f(A)/3$. We find $\hat{g}$ via a top-down descent in the circuit using the fact that $f$ is subadditive on disjoint subsets. 

In detail, we start at the unique sink $s$ of the circuit, noting that
$f(\var(s)) = f(A)$.
We maintain the invariant that $f(\var(g)) \ge f(A)/3$ for every node
$g$ in our descent. Suppose at some point we are at a node $g$. Then by assumption $f(\var(g)) \ge f(A)/3$. If also $f(\var(g)) \le 2f(A)/3$ we set $\hat{g} = g$ and are done. Otherwise, since $f(\var(g)) > 2f(A)/3$, we know that $g$ is not an input gate, so $g$ has two children $g_1, g_2$. If $g$ is a $\cup$-gate, by smoothness we have $f(\var(g_1)) = f(\var(g)) = f(\var(g_2))$, so we choose one of these children arbitrarily and continue the descent from there. Otherwise, $g$ is a $\times$-gate so by decomposability $\var(g_1)$ and $\var(g_2)$ are disjoint. Therefore, by the sub-additivity of $f$ on disjoint subsets we have $f(\var(g)) \le f(\var(g_1)) + f(\var(g_2))$. In particular, there is some child $u$ of $g$ satisfying $f(\var(u)) \ge f(\var(g))/2 \ge f(A)/3$; so we can continue the descent from $u$.
\end{proof}
}

\noindent Therefore, to prove lower bounds we will define for each $\strucA$ a structure $\strucB$ such that $\Hom(\strucA, \strucB)$ is large but the size of each $f$-balanced rectangle is small, for some appropriate function $f$.\confORfull{}{ In the next section, we will see a simple example of this where $f(S) := |S \cap W|$, where $W$ is the `highly connected set' in $\strucA$.} To upper bound the size of rectangles we use the following simple observation.

\begin{observation} \label{ob:projectub}
Let $\mathcal{R}$ be a set of functions with domain $X$. Then 
$|\mathcal{R}| \le \prod_{x \in X} |\mathcal{R}\restriction_{\{x\}}\!|.$
\end{observation}

        \section{Tight Lower Bounds for the Bounded Arity Case} \label{s:tw}

In this section, we prove Theorem~\ref{thm: tw_struc}. We first show the result for graphs and then lift it to relational structures of bounded arity. To be precise, the main part of this section is devoted to the proof of the following theorem. 

\begin{theorem} \label{thm:graph_lb}
Let  $\graphG$ be a graph. Then there exist arbitrarily large graphs $\graphH$ with $n$ vertices and at least $n^2/8$ edges such that any $\{\cup,\times\}$-circuit computing $\Hom(\graphG, \graphH)$ has size 
$\Omega\left(n^{\frac{\tw(\graphG)}{9} - (1 + \delta)}\right) $
for arbitrarily small $\delta>0$. 
\end{theorem}

\noindent To prove this result we use that every graph with high treewidth contains a set with no small balanced separator.

\begin{definition}[Balanced Separator] \label{def:bal_sep}
 Let $\graphG = (\vertexG,\edgeG)$ be a graph and $W \subseteq \vertexG$ be a non-empty set of vertices. Then we say that $S \subseteq \vertexG$ is a balanced separator with respect to $W$, if $|W \cap C| \le |W|/2$ for every connected component $C$ of $\vertexG - S$. Further, a balanced separator $S$ with respect to $W$ is a balanced $k$-separator with respect to $W$ if $|S| \le k$.
\end{definition}

\begin{lemma}[\cite{reed1997tree, robertson1995graph}] \label{lem: con}
Let $k \in \mathbb{N}$ and let $\graphG$ be a graph with $\tw(\graphG) > 3k$. Then there is some $W \subseteq \vertexG$ with $|W| = 2k+1$ having no balanced $k$-separator. 
\end{lemma}

\noindent This lemma is also used in \cite{DBLP:journals/toc/Marx10} 
to prove conditional lower bounds via fairly elaborate reductions. 
We are able to deploy Lemma~\ref{lem: con} more directly and follow a similar proof strategy to \cite{DBLP:conf/icalp/BerkholzV23}. A key difference is that we reason about rectangle covers rather than directly about the circuit. Our bad instances are the graphs given to us by the following lemma.

\begin{lemma}[{\cite[Lemma 8]{DBLP:conf/icalp/BerkholzV23}}] \label{lem: bad_graph}
For every $t \in \mathbb{N}$, there exists some $c_t \in \mathbb{R}^{+}$ such that for every sufficiently large integer $n$ there is a
  graph $\graphH = \graphH(t, n)$ with $n$ vertices, such that
  \begin{enumerate}
  \item $\graphH$ has at least $\frac18 n^2$ edges, \label{label:manyedges1}
  \item $\graphH$ contains no complete bipartite
      subgraph $\graphK_{a,a}$ for $a \geq3\log(n)$, and \label{label:nobipart1}
  \item the number of $t$-cliques in $\graphH$ is at least $c_t \cdot n^{t}$. \label{label:manycliques1}
    \end{enumerate}
  \end{lemma} 
  
\noindent So fix some graph $\graphG = (\vertexG,\edgeG)$ with $\tw(\graphG) = w$, $|\vertexG| = t$ and let $\graphH = \graphH(t,n)$ be as in Lemma~\ref{lem: bad_graph}. Our aim is to show that any $\{\cup,\times\}$-circuit computing $\Hom(\graphG, \graphH)$ has size $n^{\Omega(w)}$. The strategy is to show that any balanced rectangle $\mathcal{R} \subseteq \Hom(\graphG, \graphH)$ must be small. We do this by using the existence of a highly connected set $W$ given to us by Lemma~\ref{lem: con}. To be precise, we choose $W \subseteq \vertexG$ to be a set of size $2k+1$ having no balanced $k$-separator. 

In fact, we need to use a slight variant of the above strategy. For if $\R$ has underlying partition $(X,Y)$ with $W \subseteq X$, then it seems hard to use the properties of $W$ in a lower bound proof. 
To get around this problem we use Lemma~\ref{lem:weightedrect}, which recall is a weighted version of Lemma~\ref{lem: basic_rect}. Using this result we only have to consider rectangles which are balanced \emph{with respect to $W$}. Formally, we say a partition $(X,Y)$ of $\vertexG$ is $W$-balanced if $\min(|X \cap W|, |Y \cap W|) \ge |W|/3$. In turn, a rectangle is $W$-balanced if its underlying partition is $W$-balanced. By Lemma~\ref{lem:weightedrect}, to prove Theorem~\ref{thm:graph_lb} it suffices to show that any $W$-balanced rectangle $\mathcal{R} \subseteq \Hom(\graphG, \graphH)$ is small. The key is that there is a large matching between the two sides of any $W$-balanced partition. To formally state this we need the following notion: given disjoint subsets $X, Y \subseteq \vertexG$ an $(X,Y)$-matching (in $\graphG$) is a set of disjoint edges each of which has one endpoint in $X$ and the other in $Y$. 

\begin{lemma} \label{lem: matching}
Let $k \in \mathbb{N}$, $\graphG = (\vertexG, \edgeG)$ be a graph and $W$ be a set of $2k+1$ vertices of $\graphG$ with no balanced $k$-separator. Suppose $(X,Y)$ is a $W$-balanced partition of $\vertexG$. Then there is an $(X,Y)$-matching $\M$ in $\graphG$
of size $\lfloor k/3 \rfloor$.
\end{lemma} 

\begin{proof}
Since $(X,Y)$ is $W$-balanced we may assume w.l.o.g.\!  that 
\[ 
\frac{2k+1}{3} \le |X \cap W| \le k < |Y \cap W| \le \frac{4k+2}{3}.
\]
Further suppose, for a contradiction, that there exists a set $S' \subseteq \vertexG$ of size $\lfloor 2k/3 \rfloor$ such that removing $S'$ from $\graphG$ disconnects $X$ from $Y$. Then pick arbitrarily a set $T$ of $\lceil k/3 \rceil$ elements of $(Y \cap W) \setminus S'$ and add these to $S'$ to form a set $S$. Then $|S| = k $. We claim that $S$ is a balanced separator with respect to $W$. 

To see this, let $C$ be a connected component of $\graphG \setminus S$. Then either $C \subseteq X$ or $C \subseteq Y \setminus T$, by the assumption on $S'$. If $C \subseteq X$ then $|C \cap W| \le |X \cap W| \le k$. On the other hand if $C \subseteq Y \setminus T$ then
\begin{align*}
|C \cap W| &\le  |(Y \setminus T) \cap W| 
=|Y \cap W| - |T| 
\le \frac{4k +2}{3} - \frac{k}{3} 
= k + \frac{2}{3}
\end{align*} 
\noindent Since $|C \cap W|$ is an integer it follows that $|C \cap W| \le k$. Therefore, $S$ is indeed a balanced $k$-separator, a contradiction. So there is no such set $S'$.

To construct $\M$ we first choose any edge between $X$ and $Y$ and delete the two endpoints. We then choose an edge between $X$ and $Y$ in this new graph and repeat until $X$ and $Y$ become disconnected. By the above this leads to an $(X,Y)$-matching of size at least $\lfloor k/3 \rfloor$.
\end{proof}

\noindent By exploiting this matching it is now not too hard to show that every $W$-balanced rectangle is small; the key is that one endpoint of every edge in the matching must have small image.  

\begin{lemma} \label{lem: big_rect}
Let 
$\graphG = (\vertexG, \edgeG)$ be a $t$-vertex graph and $W \subseteq \vertexG$ be a set of vertices of size $2k+1$ with no balanced $k$-separator. Further, let $\graphH = \graphH(t,n)$ be 
the graph given by Lemma~\ref{lem: bad_graph}, for some sufficiently big 
$n$ and let $\mathcal{R} \subseteq \Hom(\graphG, \graphH)$ be a $W$-balanced rectangle. Then 
$ |\mathcal{R}| \le n^{t - \lfloor \frac{k}{3} \rfloor } \cdot [3\log(n)]^{\lfloor \frac{k}{3} \rfloor}.$
\end{lemma}

\begin{proof}
Let $(X,Y)$ be the underlying partition of $\mathcal{R}$. By Lemma~\ref{lem: matching}, there is an $(X,Y)$-matching $\M$ of size $\lfloor k/3 \rfloor$. Let $\{x,y\} \in \M$ with $x\in X, y \in Y$. Further define,
\confORfull{\begin{equation*}
    \begin{aligned}
        P_x &:= \{u \in \vertexH \mid h(x) = u, \text{ for some } h \in \mathcal{R} \}, \quad
        P_y := \{u \in \vertexH \mid h(y) = u, \text{ for some } h \in \mathcal{R} \}.
    \end{aligned}
\end{equation*}
}
{\begin{align*}
P_x &:= \{u \in \vertexH \, \mid \, h(x) = u, \textit{ for some } h \in \mathcal{R} \} \\
P_y &:= \{v \in \vertexH \, \mid \, h(y) = v, \textit{ for some } h \in \mathcal{R} \}.
\end{align*}}
Then for every $u \in P_x$ and $v \in P_y$, there is an edge in $\graphH$ between $u$ and $v$. To see this, let $h, g \in \mathcal{R}$ with $h(x) = u$ and $g(y) = v$. Hence 
$ h\! \restriction_{X} \times \, g\!\restriction_{Y} \, \in \mathcal{R} \subseteq \Hom(\graphG, \graphH)$,
 since $\mathcal{R}$ is an $(X,Y)$-rectangle. Therefore, since $\{x,y\} \in \edgeG$, we have that $\{h(x), g(y)\} = \{u, v\} \in \edgeH$. It follows that either $|P_x| \le 3\log(n)$ or $|P_y| \le 3\log(n)$, since $\graphH$ contains no $\graphK_{3\log(n), 3\log(n)}$ subgraph by Lemma~\ref{lem: bad_graph}(\ref{label:nobipart1}). Moreover, by definition $|P_{x}| = |\mathcal{R}\!\restriction_{\{x\}}\!|$ and $|P_{y}| = |\mathcal{R}\!\restriction_{\{y\}}\!|$.

Since the above holds for every $\{x,y\} \in \M$, there are at least $\lfloor k/3 \rfloor$ vertices $v \in \vertexG$ such that  $|\mathcal{R}\!\restriction_{\{v\}}\!| \le 3\log(n)$. The result follows by Observation~\ref{ob:projectub}.
\end{proof}

\confORfull{\noindent To infer Theorem~\ref{thm:graph_lb}, first observe that by Lemma~\ref{lem: bad_graph}(\ref{label:manycliques1}), $|\!\Hom(\graphG, \graphH)| = \Omega(n^{t})$. Here we use the fact that any mapping $h \colon \vertexG \to \vertexH$ whose image is a clique is a homomorphism. Therefore, by Lemmas~\ref{lem:weightedrect} and \ref{lem: big_rect} the result follows; see Appendix~\ref{a:tw} for full details.
}
{\noindent We can now prove our lower bound.}

\confORfull{}{
        \begin{proof}[Proof of Theorem~\ref{thm:graph_lb}]

	Fix $\delta >0$, suppose that $\graphG$ has $t$ vertices and let $\graphH = \graphH(t, n)$ be the graph given by Lemma~\ref{lem: bad_graph}. Firstly, if $w:= \tw(\graphG) \le 18$ then the bound is trivial; this follows since any $\{\cup, \times\}$-circuit computing $\Hom(\graphG, \graphH)$ must have $\Omega(n)$ input gates by Observation~\ref{ob:projectub} and Lemma~\ref{lem: bad_graph}(\ref{label:manycliques1}). It remains to prove the bound for $w \ge 19$. 
	
	In this case let $3k$ be the maximum multiple of three such that $w > 3k$; then $k \ge 6$ and $3k+3 \ge w$. By Lemma~\ref{lem: con}, we know that there is some $W \subseteq \vertexG$ with $|W| = 2k+1$ having no balanced $k$-separator. Therefore, by Lemma~\ref{lem: big_rect}, every $W$-balanced rectangle $\mathcal{R}$ with $\mathcal{R} \subseteq \Hom(\graphG, \graphH)$ has size at most 
	\[
	n^{t - \lfloor \frac{k}{3} \rfloor } \cdot [3\log(n)]^{\lfloor \frac{k}{3} \rfloor}. 
	\]
	On the other hand, the number of $t$-cliques in $\graphH$ is at least $c_{t} n^{t}$ by Lemma~\ref{lem: bad_graph}(\ref{label:manycliques1}). Therefore, there are at least $t!c_{t} n^{t}$ homomorphisms from $\graphG$ to $\graphH$, since any mapping $h \colon \vertexG \to \vertexH$ whose image is a clique is a homomorphism. Therefore, applying Lemma~\ref{lem:weightedrect}, we obtain that any $\{\cup,\times\}$-circuit computing $\Hom(\graphG, \graphH)$ must have size at least
	\[
	\frac{t!c_{t} n^{t}}{n^{t- \lfloor \frac{k}{3} \rfloor} \cdot [3\log(n)]^{ \lfloor \frac{k}{3} \rfloor}} = \frac{t!c_{t} n^{ \lfloor \frac{k}{3} \rfloor}}{[3\log(n)]^{ \lfloor \frac{k}{3} \rfloor}}.
	\]
	The result follows by the maximality of $k$. In detail, since $3k+3 \ge w$ we obtain that 
	\[ 
	\left\lfloor \frac{k}{3} \right\rfloor \ge \frac{w}{9} - 1 > \frac{w}{9} - (1 + \delta)
	\] 
	The asymptotic bound follows by taking $n$ to be sufficiently big.
\end{proof} %
}

\confORfull{}{\noindent}We next lift this to relational
structures. We could do this similarly to the above. However, it is
convenient to instead use the following reduction result from
\cite{DBLP:conf/icalp/BerkholzV23}. \confORfull{}{Theorem~\ref{thm: tw_struc} then
follows immediately.}

\begin{lemma}[{\cite[Lemmas~14-16]{DBLP:conf/icalp/BerkholzV23}}] \label{lem:reduct}
        Let $a, r \in \mathbb{N}$ and let $\strucA$ be a $\sigma$-structure with arity at most $r$. Let $\graphG$ be the Gaifman graph of $\strucA$ and suppose that there exists a graph $\graphH$ such that any $\{\cup,\times\}$-circuit computing $\Hom(\graphG, \graphH)$ has size at least $\|\graphH\|^{a}$. Then there is a $\sigma$-structure $\strucB$ with $\|\strucB\| \ge \|\graphH\|$ and a function $g$ such that any $\{\cup,\times\}$-circuit computing $\Hom(\strucA, \strucB)$ has size at least $g(\|\strucA\|) \cdot \|\strucB\|^{2a/r}$.
\end{lemma}

\begin{proof}[Proof of Theorem~\ref{thm: tw_struc}]
        Let $\graphG$ be the Gaifman graph of $\strucA$  and let $N'$ be an integer. Then since 
        the treewidth of $\graphG$ is the same as the treewidth of $\strucA$ we may apply 
        Theorem~\ref{thm:graph_lb}, to obtain a graph $\graphH$ with $n \ge N'$ vertices 
 and at least $n^2/8$ many edges such that any $\{\cup,\times\}$-circuit computing 
 $\Hom(\graphG, \graphH)$ has size $\Omega(n^{\frac{w}{9}-2}) = 
 \|\graphH\|^{\Omega(w)}$. Therefore, by Lemma~\ref{lem:reduct}, there is a 
 $\sigma$-structure $\strucB$ with $\|\strucB\| \ge \|\graphH\| \ge N'$ such that any 
 $\{\cup,\times\}$-circuit for $\Hom(\strucA, \strucB)$ has size $\|\strucB\|^{\Omega(w)}$, 
 since $r$ is a constant and by choosing $n$ (and hence $\|\strucB\|$) to be sufficiently 
 big.
\end{proof}

\noindent Note that some form of dependence on the arity is necessary since there are classes of structures $\classA$ of unbounded treewidth such that solving $\{\cup,\times\}$-$\Hom(\classA)$ is in polynomial time, e.g. classes of unbounded treewidth but bounded fractional hypertree width. Such classes must have unbounded arity; we next turn to this setting.

\section{Submodular Width Lower Bound} \label{s:sw}
In this section we prove Theorem~\ref{thm: sub_lb}. Our strategy is analogous to the one used in Section~\ref{s:tw}: we use the fact that high submodular width implies the existence of a `highly connected set'. However, the exact notion of highly connectedness is more complicated and we need to define our bad instances more carefully. As such we first make some preparations.

\subsection{Preliminaries to the Proof} \label{ss:prelim_subw}

\confORfull{F}{So f}ix a $\sigma$-structure $\strucA$ and let $\graphH =
\graphH(\strucA) = (\vertexH, \edgeH)$ be the hypergraph of $\strucA$. We use similar notation to
  \cite{DBLP:journals/jacm/Marx13}. 
A \emph{path} $P$ in $\graphH$ is a sequence of vertices forming a path in the Gaifman graph of $\graphH$. The \emph{endpoints} of the path are the first and last vertex in this sequence. For $A, B \subseteq \vertexH$ we say a path is an $(A,B)$-path if it starts in $A$ and ends in $B$. Let $\mathcal{P}$ denote the set of paths in $\graphH$. A \emph{flow} of $H$ is an assignment $F \colon \mathcal{P} \to \mathbb{R}^{\ge 0}$ such that for every $e \in \edgeH$,
\[
\sum_{\substack{P \in \mathcal{P} \\ P \cap e \neq \emptyset}} F(P) \le 1.
\]
If $F(P) > 0$, then we say that $P$ appears in $F$ or is a path of $F$. The value of $F$ is $\sum_{P \in \mathcal{P}} F(P)$. For $A, B \subseteq \vertexH$ a flow $F$ is an $(A,B)$-flow if only $(A,B)$-paths appear in $F$. If $F_1, \dots, F_k$ are flows then their sum is the mapping which assigns weight $\sum_{i=1}^k F_i(P)$ to each path $P$. If this mapping happens to also be a flow then we say that the $F_i$ are compatible. 

\confORfull{}{The following gives a characterisation of what it means to have high submodular width.

\begin{theorem}[{\cite[Theorem 5.1.]{DBLP:journals/jacm/Marx13}}] \label{thm: sub_char}
Let $\lambda >0$ be a constant and $\strucA$ be a relational structure with $\subw(\strucA) > w$ and hypergraph $\graphH = (\vertexH, \edgeH)$. Then there exists a fractional independent $\mu$ of $\graphH$ and a set $W \subseteq \vertexH$ such that:
\begin{enumerate}
\item $\mu(W) \ge \frac{2}{3} w -1$, and
\item for every disjoint $A, B \subseteq W$ there is an $(A,B)$-flow of value $\lambda \cdot \min(\mu(A), \mu(B))$.
\end{enumerate}  
\end{theorem}

\noindent The problem with the above result is that the flow depends on the exact subsets $A, B \subseteq W$ under consideration. We want to use these flows to define our bad instances, but this is difficult to do since we do not have one fixed flow but many flows. To get around this we use a different characterisation of submodular width.}

\confORfull{For the characterisation of submodular width we will deploy the following notion is crucial.}{For this the following notion is crucial.} Let $X_1, \dots, X_k$ be subsets of $\vertexH$. Then a uniform concurrent flow of value $\epsilon$ on $(X_1, \dots, X_k)$ is a collection of compatible flows $F_{i,j}$, with $1 \le i < j \le k$ such that $F_{i,j}$ is an $(X_i, X_j)$ flow of value exactly $\epsilon$. If $\bigcup X_i$ is connected in $\graphH$ then by choosing one path connecting each pair $(X_i, X_j)$ and assigning this path weight $1/\binom{k}{2}$ we obtain a uniform concurrent flow with value $\Omega(k^{-2})$. The next theorem says that if $\graphH$ has high submodular width then there is a `large' set which can be partitioned into cliques admitting a uniform concurrent flow with value beating the trivial $k^{-2}$ lower bound. This characterisation is obtained and used by Marx in \cite{DBLP:journals/jacm/Marx13} and follows from Theorem~5.1 and Lemmas~6.5 and 6.6 in that paper.

\begin{theorem}[\cite{DBLP:journals/jacm/Marx13}] \label{thm: sub_con}
Let $\strucA$ be a relational structure with $\subw(\strucA) = w$ and hypergraph $\graphH = (\vertexH, \edgeH)$. Then there exists a set $W \subseteq \vertexH$ which can be partitioned into $(K_1, \dots, K_k)$ where each $K_i \subseteq e$ for some $e \in \edgeH$ and $k  = \Omega(\sqrt{w})$, such that there is a uniform concurrent flow of value $\Omega(k^{-3/2})$ on $(K_1, \dots, K_k)$.  
\end{theorem}

\noindent Marx uses these flows to show that hypergraphs of high submodular width have high `embedding power'. 
We argue differently and more directly. 

\subsection{Lower Bound}

\subsubsection{Simplifying the Structures} \label{ss:simple}

\noindent In this section we show that in order to prove Theorem~\ref{thm: sub_lb} it suffices to prove lower bounds for left-hand-side structures $\strucA$ of a simple form. \confORfull{This is straightforward so we just describe the simplification here and defer the proof of correctness until Appendix~\ref{ass:simp}.}{To do this we use a number of reduction results, similar to Lemma~\ref{lem:reduct}.}

        \confORfull{
         We first need some definitions. Let $\tau$ be a signature and $\strucX, \strucY$ be $\tau$-structures. We say that $\strucY$ is \emph{coordinate respecting (relative to $\strucX$)} if for every $h, h' \in \Hom(\strucX, \strucY)$ and every distinct $x, x' \in X$ it holds that $h(x) \neq h'(x')$. In this case we define 
         \[ 
         \dom_{\strucY}(x) := \{y\in Y \, \mid \, h(x) = y \textit{ for some } h \in \Hom(\strucX, \strucY)\}.
         \]
         Note that $\dom_{\strucY}(x) \cap \dom_{\strucY}(x')$ is empty for all $x \neq x'$. 
         Since $\strucY$ is normally clear from the context we often write $\dom(x)$ for $\dom_{\strucY}(x)$.
         
         Next let $\strucX$ be a $\tau$-structure such that every relation contains exactly one tuple. Let $\leq$ be a total order on $X$ and 
$\bar{t}$ be some $r$-tuple with coordinates in $X$. Moreover, suppose
$\bar{t}$ has $i$ distinct coordinates. Then we define
$\bar{t}_{\leq}$ to be the unique $i$-tuple with $\tilde{t} =
\tilde{t}_{\leq}$ whose coordinates are ordered by $\leq$. Then if
$\bar{t}$ is the unique tuple in $R^{\strucX}$ and if $\bar{t}_{\leq}$
is an $i$-tuple we define $R_{\leq}$ to be a fresh $i$-ary relation
symbol and define $\tau_{\leq}:= \{R_{\leq} \, \mid \, R \in
\sigma\}$. Now we define the $\tau_{\leq}$-structure $\strucX^{\leq}$
to have universe $X$ and such that $R_{\leq}^{\strucX^{\leq}} =
\{\bar{t}_{\leq}\}$ where $\bar{t}$ is the unique tuple in
$R^{\strucX}$. Moreover, we say that $\strucX$ is \emph{order respecting} if $\strucX = \strucX^{\leq}$ for some total order $\leq$. 
         
        Let $\strucY$ be a  $\tau_{\leq}$-structure which is coordinate respecting 
         relative to $\strucX^{\leq}$. Then $\leq$ induces a partial order $\leq_Y$ on $Y$ as 
         follows. For $y_1, y_2 \in Y$, $y_1 \leq_Y y_2$ iff  there are distinct $x_1, x_2 \in 
         X$ with $y_1 \in \dom(x_1)$, $y_2 \in \dom(x_2)$ and $x_1 \leq x_2$. We say that $
         \strucY$ is order respecting relative to $\strucX^{\leq}$ if for every $R \in \tau_{\leq}$ and every $\bar{t} 
         \in R_{\leq}^{\strucY}$, we have that $t_i \leq_{Y} t_j$ for all $i<j$.     
         
To prove Theorem~\ref{thm: sub_lb} it suffices to prove the following result, see Appendix~\ref{ass:simp} for details. 
         
         \begin{theorem} \label{thm:sub_lb_reduced}
                Let $\strucA$ be an order respecting structure where every relation contains exactly one tuple. Then there exist arbitrarily large structures $\strucB$ which are coordinate and order respecting relative to $\strucA$ and such that any $\{\cup,\times\}$-circuit computing $\Hom(\strucA, \strucB)$ has size $\|\strucB\|^{\Omega(\subw(\strucA)^{1/4})}$. 
         \end{theorem}
         
         \noindent In our proof of Theorem~\ref{thm:sub_lb_reduced}, we often want to go between sets and tuples. Of course, for every tuple $\bar{a} = (a_1, \dots, a_r)$ there is a naturally corresponding set $\tilde{a} := \{a_1, \dots, a_r\}$. But in the other direction given a set there are many possible corresponding tuples.
         This is why it is nice to work with $
         (\strucA, \strucB)$ respecting an order $\leq$ on $A$. Now for any
         $S \subseteq A$, we define $\overrightarrow{S}$ to be the 
         unique tuple respecting the order $\leq$ whose coordinates are the elements of $S$. We can similarly deploy this notation 
         for any subset of $B$ which is `colourful', in the sense that 
         no two of its elements appear in the same $\dom(a)$ set. 
}{}

\confORfull{}{We first carry out a simple trick, which we also deployed in \cite{DBLP:conf/icalp/BerkholzV23}. For this we introduce the \emph{individualisation} of a $\sigma$-structure $\strucA$, which is obtained by giving every element in the universe a unique colour. To be precise we extend $\sigma$ with a unary relation $P_a$ for every $a \in A$. Let $\sigma_{\!A} = \{P_a \, \mid \, a \in A\}$ and let $\strucA^{\textrm{id}}$ be the $\sigma \cup \sigma_{\!A}$ expansion of $\strucA$ defined by $P_a^{\strucA^{\textrm{id}}} = \{a\}$. We call a structure \emph{individualised} if every element in the universe has a unique colour. The next lemma shows that it suffices to prove lower bounds for individualised structures, provided the elements of the universe of the `bad instances' are partitioned by $\sigma_{\!A}$. 

\begin{lemma}[{\cite[Lemma 14]{DBLP:conf/icalp/BerkholzV23}}] \label{lem:any_to_individual}
        Let $\strucA$ be a $\sigma$-structure and let $\strucB'$ be a $\sigma \cup \sigma_{\!A}$ structure such that $\{P_a^{\strucB} \, \mid \, a \in A\}$ is a partition of the universe. Let $\strucB$ be the $\sigma$-reduct of $\strucB$. Then if there is a $\{\cup,\times\}$-circuit computing $\Hom(\strucA, \strucB)$ of size $s$, then there is a $\{\cup,\times\}$-circuit computing $\Hom(\strucA^{\textrm{id}}, \strucB')$ of size at most $s$.
        \end{lemma} 

\noindent Our next step reduces from individualised structures to structures where every  relation contains only a single tuple, provided that our bad instances have certain properties. To this end we define from $\strucA$---an individualised $\sigma$-structure---a related structure $\strucA^{\textrm{sp}}$ where every relation contains only a single tuple. To this end, define $\sigma_{\!A}:= \{P_a \, \mid \, a\in A\} \subseteq \sigma$ to be the unary relations in $\sigma$ which give each element a unique colour. Further, let
\[ 
\sigma^{\textrm{sp}} := \{R_{\bar{t}} \, \mid \, R \in \sigma \setminus \sigma_{\!A}, \, \bar{t} \in R^{\strucA} \}.
 \]
Then we define $\strucA^{\textrm{sp}}$ to be the $\sigma^{\textrm{sp}} $-structure with universe $A$ such that for every $R \in \sigma \setminus \sigma_{\!A}$ and every $\bar{t} \in R^{\strucA}$ we have that $R_{\bar{t}}^{\strucA^{\textrm{sp}}} = \{\bar{t}\}$. 

Before stating our next reduction lemma we need to introduce the two properties required of our bad instances. So let $\strucX, \strucY$ be $\tau$-structures, for some signature $\tau$. Then we say that $\strucY$ is \emph{coordinate respecting (relative to $\strucX$)} if for every $h, h' \in \Hom(\strucX, \strucY)$ and every distinct $x, x' \in X$ it holds that $h(x) \neq h'(x')$. In this case we define 
\[ 
\dom_{\strucY}(x) := \{y\in Y \, \mid \, h(x) = y \textit{ for some } h \in \Hom(\strucX, \strucY)\}.
\]
Note that $\dom_{\strucY}(x) \cap \dom_{\strucY}(x')$ is empty for all $x \neq x'$. 
When $\strucY$ is clear from the context we write $\dom(x)$ instead of $\dom_{\strucY}(x)$.

We need one further concept in the proof. We say that $\strucY$ is \emph{reduced (relative  to $\strucX$)} if its relations contain no tuples which are irrelevant to $\Hom(\strucX, \strucY)$. Formally, this means that for every $R \in \sigma$ and $\bar{y} \in R^{\strucY}$ there exists $\bar{x} \in R^{\strucX}$ and $h \in \Hom(\strucX, \strucY)$ with $h(\bar{x}) = \bar{y}$. If $(R, \bar{y})$ violates this property, then we may remove $\bar{y}$ from $R^{\strucY}$ without affecting $\Hom(\strucX, \strucY)$. This leads to the following observation.

\begin{observation} \label{ob:reduced}
For any $\sigma$-structure $\strucY$, there exists a reduced structure $\strucY'$ with $\|\strucY'\| \le \|\strucY\|$ such that $\Hom(\strucX, \strucY) = \Hom(\strucX, \strucY')$. 
\end{observation}

\noindent We now formally state the next step of our reduction.

\begin{lemma} \label{lem:individual_to_sparse}
Let $\strucA$ be an individualised $\sigma$-structure and let $\strucD$ be a $\sigma^{\textrm{sp}}$-structure which is coordinate respecting and reduced relative to $\strucA^{\textrm{sp}}$. Then there exists a $\sigma$-structure $\strucB$ with $\|\strucB\| = \|\strucD\| + \sum_{a \in A} |\dom_{\strucD}(a)|$ such that if $\Hom(\strucA, \strucB)$ admits a {\ctc} of size $s$ then so does $\Hom(\strucA^{\textrm{sp}}, \strucD)$.
\end{lemma}

\begin{proof}
We define $\strucB$ to have the same universe as $\strucD$. Then for each $a \in A$ and $R \in \sigma \setminus \sigma_{\! A}$ we define
\begin{align*}
R^{\strucB} &:= \bigcup_{\bar{t} \in R^{\strucA}} R_{\bar{t}}^{\strucD} \\
P_a^{\strucB} &:= \dom_{\strucD}(a) 
\end{align*}
Then $\|\strucB\| = \|\strucD\| + \sum_{a\in A} |\dom_{\strucD}(a)|$. Therefore, to prove the lemma it suffices to show that $\Hom(\strucA, \strucB) = \Hom(\strucA^{\textrm{sp}}, \strucD)$. 

First, let $h \in \Hom(\strucA, \strucB)$ and consider a 
relation $R_{\bar{t}} \in \sigma^{\textrm{sp}}$ with arity $r$. Then the 
unique tuple in this relation in $\strucA^{\textrm{sp}}$ is $\bar{t}$ which is also in $R^{\strucA}$. Since $h 
\in \Hom(\strucA, \strucB)$ it follows that $h(\bar{t}):= (h(t_1), \dots, h(t_{r})) 
\in R^{\strucB}$ and that $h(t_i) \in P_{t_i}^{\strucB}$ for all $i 
\in [r]$. By the definition of $R^{\strucB}$, we obtain that $h(\bar{t}) \in R^{\strucD}_{\bar{s}}$ for some $\bar{s} \in R^{\strucA}$. Then since $\strucD$ is reduced relative to $\strucA^{\textrm{sp}}$, there is some $\bar{a} \in R_{\bar{s}}^{\strucA^{\textrm{sp}}}$ and a homomorphism $h' \in \Hom(\strucA^{\textrm{sp}}, \strucD)$ with $h'(\bar{a}) = h(\bar{t})$. But since $R_{\bar{s}}^{\strucA^{\textrm{sp}}} = \{\bar{s}\}$, it follows that $\bar{a} = \bar{s}$. This implies that $h(t_i) \in \dom_{\strucD}(s_i)$ for all $i \in [r]$. But also $h(t_i) \in P_{t_i}^{\strucB} = \dom_{\strucD}(t_i)$. Since $\strucD$ is coordinate respecting it follows that $t_i = s_i$ for all $i \in [r]$, so $h(\bar{t}) \in R^{\strucD}_{\bar{t}}$ and $h \in \Hom(\strucA^{\textrm{sp}}, \strucD)$. 

Conversely, let  $h \in \Hom(\strucA^{\textrm{sp}}, \strucD)$ and 
let $\bar{t} \in R^{\strucA} $ for some $R \in \sigma \setminus \sigma_{\!A}$. 
Then by construction $\bar{t} \in R^{\strucA^{\textrm{sp}}}_{\bar{t}}$. Therefore, since $h \in \Hom(\strucA^{\textrm{sp}}, \strucD)
$, we obtain $h(\bar{t}) \in R^{\strucB}_{\bar{t}} \subseteq R^{\strucB}$. Therefore, $h$ respects every relation in $\sigma \setminus \sigma_{\!A}$. Moreover, for all $i \in [r]$ we have $h(t_i) \in \dom_{\strucD}(t_i) = P_{t_i}^{\strucB}$ by the definition of $\dom_{\strucD}(t_i)$. The result follows, since $\strucA$ is connected by assumption. In detail, for every $a \in A$ there exists---by connectedness---some $\bar{t} \in R^{\strucA} \in \sigma \setminus \sigma_A$ such that $t_i = a$ for some $i \in [r]$. Therefore, $h(a) \in P_{a}$ and so $h$ also respects every relation in $\sigma_{\!A}$.
\end{proof}

\noindent We end this subsection by making a further simplification designed to circumvent a technical annoyance. Namely, that in our proofs we often want to go between sets and tuples. Of course, for every tuple $\bar{a} = (a_1, \dots, a_r)$ there is a naturally corresponding set $\tilde{a} := \{a_1, \dots, a_r\}$. But, in the other direction, given a set there are many possible corresponding tuples, and it can be messy to specify which we mean in our proofs. The following simplifications are therefore not strictly necessary but save us some notational headaches later on. To get to this point we need to jump through some technical hoops but everything we do is easy; just a little annoying to formally state.

Let $\strucX$ be a $\tau$-structure such that every relation
contains exactly one tuple. Let $\leq$ be a total order on $X$, the
universe of $\strucX$. Let $\bar{t}$ be some $r$-tuple with
coordinates in $X$. Moreover, suppose $\bar{t}$ has $i$ distinct
coordinates. Then we define $\bar{t}_{\leq}$ to be the unique
$i$-tuple with $\tilde{t} = \tilde{t}_{\leq}$ whose coordinates are
ordered by $\leq$. Then if $\bar{t}$ is the unique tuple in
$R^{\strucX}$ and if $\bar{t}_{\leq}$ is an $i$-tuple we define
$R_{\leq}$ to be a fresh $i$-ary relation symbol and define
$\tau_{\leq}:= \{R_{\leq} \, \mid \, R \in \sigma\}$. Now we define
the $\tau_{\leq}$-structure $\strucX^{\leq}$ to have universe $X$ and
such that $R_{\leq}^{\strucX^{\leq}} = \{\bar{t}_{\leq}\}$ where
$\bar{t}$ is the unique tuple in $R^{\strucX}$. Moreover, we say that
$\strucX$ is \emph{order respecting} if $\strucX = \strucX^{\leq}$ for some total order $\leq$. 

 Let $\strucY$ be a  $\tau_{\leq}$-structure which is coordinate respecting 
         relative to $\strucX^{\leq}$. Then $\leq$ induces a partial order $\leq_Y$ on $Y$ as 
         follows. For $y_1, y_2 \in Y$, $y_1 \leq_Y y_2$ iff  there are distinct $x_1, x_2 \in 
         X$ with $y_1 \in \dom(x_1)$, $y_2 \in \dom(x_2)$ and $x_1 \leq x_2$. We say that $
         \strucY$ is order respecting relative to $\strucX^{\leq}$ if for every $R \in \tau_{\leq}$ and every $\bar{t} 
         \in R_{\leq}^{\strucY}$, we have that $t_i \leq_{Y} t_j$ for all $i<j$. We can now state our 
final simplification lemma; note that we now need to also assume that our bad instances 
respect an order.

\begin{lemma} \label{lem:order}
Let $\strucA$ be a $\sigma$-structure where every relation contains exactly one tuple, let $\leq$ be a total order on $A$ and let $\strucD$ be a $\sigma_{\leq}$-structure which is reduced, coordinate respecting, and order respecting relative to $\strucA^{\leq}$. Then there exists a $\sigma$-structure $\strucB$ with $\|\strucB\| = \|\strucD\|$ which is reduced and coordinate respecting relative to $\strucA$, and such that if $\Hom(\strucA,\strucB)$ admits a {\ctc} of size $s$ then so does $\Hom(\strucA^{\leq}, \strucD)$. 
\end{lemma}

\begin{proof}
We define $\strucB$ to have the same universe as $\strucD$ and define its relations by---in some sense---reversing the tuple-ordering process. Formally, for each $r$-ary relation $R \in \sigma$ let $\bar{t}$ be the unique $r$-tuple in $R^{\strucA}$. Suppose that $R^{\strucA^{\leq}}$ is an $i$-ary relation. By construction the unique tuple in  $R^{\strucA^{\leq}}_{\leq}$ is  $\bar{t}_{\leq}$. Then 
 define $f \colon [r] \to [i]$ to be the map such that the $j$-th coordinate of $\bar{t}$ is equal to the $f(j)$-th coordinate of $\bar{t}_{\leq}$. Then for each $i$-tuple $\bar{s} \in R^{\strucD}_{\leq}$, we add an $r$-tuple to $R^{\strucB}$ whose $j$-th coordinate is $s_{f(j)}$ (recall this denotes the $f(j)$-th coordinate of $\bar{s}$). The fact that $\strucB$ is reduced and coordinate respecting relative to $\strucA$ follows from the fact that $\strucD$ is reduced and order respecting relative to $\strucA^{\leq}$. Moreover, it is easy to verify that $\Hom(\strucA, \strucB) = \Hom(\strucA^{\leq}, \strucD)$ and so the result follows.
\end{proof}

\noindent Why is this simplification useful? Well, if $\strucA$ is
order respecting, then for any
$S \subseteq A$, we define $\overrightarrow{S}$ to be the 
unique tuple respecting the order $\leq$ whose coordinates are the elements of $S$. Moreover, if $\strucB$ is order respecting relative to $\strucA$, then we can similarly deploy this notation 
for any subset of $B$ which is `colourful', in the sense that 
no two of its elements appear in the same $\dom(a)$ set. 

To summarise, we have shown that in order to prove Theorem~\ref{thm: sub_lb} it suffices to prove the following. 

\begin{theorem} \label{thm:sub_lb_reduced}
Let $\strucA$ be an order respecting structure where every relation contains exactly one tuple. Then there exist arbitrarily large structures $\strucB$ which are coordinate and order respecting relative to $\strucA$ and such that any $\{\cup,\times\}$-circuit computing $\Hom(\strucA, \strucB)$ has size $\|\strucB\|^{\Omega(\subw(\strucA)^{1/4})}$. 
\end{theorem}

\noindent To see how this implies Theorem~\ref{thm: sub_lb}, let $\strucA$ be an 
arbitrary $\sigma$-structure. Now let $\leq$ be an arbitrary order on $A$ and  consider 
$\strucA' := ((\strucA^{\textrm{id}})^{\textrm{sp}})^{\leq}$. Then by construction $
\strucA'$ is an order respecting structure where every relation contains exactly one 
tuple. Therefore, we can apply Theorem~\ref{thm:sub_lb_reduced} to obtain an arbitrarily 
large structure $\strucB'$ which is  coordinate and order respecting relative to $\strucA$ and such that any $\{\cup,\times\}$-circuit computing 
$\Hom(\strucA', \strucB')$ has size $\|\strucB'\|^{\Omega(\subw(\strucA)^{1/4})}$. Further, we can apply Observation~\ref{ob:reduced} to obtain a reduced, coordinate respecting, and order respecting structure such that any $\{\cup,\times\}$-circuit computing 
$\Hom(\strucA', \strucB'')$ has size $\|\strucB''\|^{\Omega(\subw(\strucA)^{1/4})}$. Moreover, by taking $\strucB'$ big enough we can ensure $ \strucB''$ is arbitrarily big. Therefore, we can first apply Lemma~\ref{lem:order}, then 
Lemma~\ref{lem:individual_to_sparse} and finally Lemma~\ref{lem:any_to_individual} to obtain a 
$\sigma$-structure $\strucB$ with $\|\strucB\| \ge \|\strucB'\|$ such that any $\{\cup,
\times\}$-circuit computing $\Hom(\strucA, \strucB)$ has size $\|\strucB\|
^{\Omega(\subw(\strucA)^{1/4})}$; which implies Theorem~\ref{thm: 
sub_lb}. It remains to prove  Theorem~\ref{thm:sub_lb_reduced}.} 

\subsubsection{Defining the Bad Instances}
 
Towards proving Theorem~\ref{thm:sub_lb_reduced}, let $\sigma$ be a signature and let $\strucA$ be an order respecting $\sigma$-structure such that every relation in $\strucA$ contains exactly one tuple. Let $\graphH = (\vertexH, \edgeH)$ be the hypergraph of $\strucA$. Further suppose $\subw(\strucA) = w$ so that, by Theorem~\ref{thm: sub_con}, there exists some $W \subseteq \vertexH$ which can be partitioned into $(K_1, \dots, K_k)$ for some $k = \Omega(\sqrt{w})$, where each $K_i$ is a clique such that there is a uniform concurrent flow $F = \{ F_{i,j} \, \mid \, 1 \le i < j \le k\}$ of value $\varepsilon = \Omega(k^{-3/2})$ on $(K_1, \dots, K_k)$. The aim of this subsection is to define appropriate bad instances $\strucB$. Analogously to in the treewidth case we will use random instances. However, we need to carefully control how each vertex contributes to the number of homomorphisms from $\strucA$. A crucial idea is to use the flow $F$ to do this. Essentially we use $F$ to define a fractional independent set which then allows us to define an instance $\strucB$ with relations bounded by some constant $N$. This is similar in spirit to the upper bound of the `AGM bound' \cite[Lemma 4]{DBLP:conf/focs/AtseriasGM08} but we also sprinkle in some randomness to ensure hardness. 

Concretely, let $N$ be a sufficiently big integer, $\graphH = (\vertexH, \edgeH)$ be the hypergraph of $\strucA$ and for $v \in \vertexH$, let $\mathcal{P}_v$ be the set of 
paths in $\graphH$ that contain the vertex $v$. Define a map $\mu 
\colon \vertexH \to [0,1]$ by $\mu(v) := \frac{1}{2}\sum_{P \in 
\mathcal{P}_{v}} F(P)$. Then 
for each $v \in \vertexH$ we define disjoint sets of elements $
\dom(v)$ such that $|\!\dom(v)| = N^{\mu(v)}$. The idea is that the size of $\log(|\!\dom(v)|)$ is proportional to the weight of paths in $F$ passing through $v$. We define the 
\emph{random structure $\strucB$ induced by $F$ with relations bounded by $N$} as follows. For each 
relation $R \in \sigma$, let $\bar{v} = (v_1, \dots, v_r)$ be the unique tuple 
in $R^{\strucA}$. Then $R^{\strucB}$ is a random subset of $
\bigtimes_{i=1}^r \dom(v_i)$, where each tuple is included 
independently with probability a half. Note that,
\[
|R^{\strucB}| \le \prod_{i=1}^r|\!\dom(v_i)| = N^{\alpha}, \quad \text{where } \alpha := \sum_{i=1}^r \mu(v_i).
\]
Moreover, $\tilde{v}$ is an edge of $\graphH$ by definition. Since $F$ is a flow we obtain that 
\[
\sum_{v \in e} \mu(v) = \frac{1}{2}\sum_{v \in e} \sum_{P \in \mathcal{P}_{v}} F(P) \le  \sum_{\substack{P \in \mathcal{P}, P \cap e \neq \emptyset}} F(P) \le 1.
\] 
Here we assume w.l.o.g.\! that each path appearing in $F$ is minimal and therefore intersects each edge at most twice. Using this (non-standard) model of a random structure we can obtain an analogue to Lemma~\ref{lem: bad_graph}. There we showed that our bad instances did not contain any large complete bipartite subgraphs. The natural suggestion would be to look at the Gaifmann graph of $\strucB$ and show that it (with high probability) has no large complete bipartite subgraphs. However, the Gaifmann graph of $\strucB$ is very dense so this does not work. Instead the following definition gives us the analogous `badness' property in the relational setting. 

\begin{definition} \label{def:scattered}
Let $N$ be a positive integer, $\strucB$ a $\sigma$-structure and $R \in \sigma$ with $R^{\strucB} \subseteq\bigtimes_{i \in [r]} \dom(v_i)$ for some $\bar{v}$ an $r$-tuple of vertices of $\graphH$. Then we say that $R$ is \emph{$N$-scattered} if for every partition of $\tilde{v}$ $(U,V)$, the following holds. For every $S \subseteq \bigtimes_{v \in U} \dom(v)$ and $T \subseteq \bigtimes_{v \in V} \dom(v)$, such that 
\[
\left\{\overrightarrow{\left(\tilde{s} \cup \tilde{t} \right)} \, \mid \, s \in S, \, t \in T \right\} \subseteq R^{\strucB} 
\]
either $|S| \le 3\log(N)$ or $|T| \le 3\log(N)$. 
\end{definition}

\noindent The following lemma shows that the bad instances we need exist\confORfull{; it is proved by a simple probabilistic argument, see Appendix~\ref{ass:rest}.}{.} 

\begin{lemma} \label{lem: bad_hypergraph}
There exists a constant $c = c(\|A\|)$ such that for every sufficiently large integer $N$ there exists a $\sigma$-structure $\strucB = \strucB(N)$ which is coordinate and order respecting relative to $\strucA$. Moreover, $\strucB$ has universe $\bigcup_{v \in \vertexH} \dom(v)$,  size $\|\strucB\| \le \|A\| \cdot N$ and satisfies the following properties:
\begin{enumerate}
\item $|\!\Hom(\strucA, \strucB)| \ge cN^{t}$, where $t := \sum_{v \in \vertexH} \mu(v)$ and \label{big_instance}
\item for every $R \in \sigma, R^{\strucB}$ is $N$-scattered. \label{nonuni}
\end{enumerate}
\end{lemma}

\confORfull{}{
\begin{proof}
	We proceed by showing that the probability that $\strucB$, the random structure induced by $F$ with relations bounded by $N$, satisfies the lemma is greater than zero. We begin with the following claim; although the setting is different the core probability calculation is the same as in the proof of Lemma~\ref{lem: bad_graph}. 
	\begin{claim}
		Let $N>0$, $R \in \sigma$ and $\strucB$ be the random structure induced by $F$ with relations bounded by $N$. Then the probability that $R$ is $N$-scattered tends to one as $N \to \infty$. 
	\end{claim}
	
	\begin{claimproof}
		Suppose that $R^{\strucA} = \{\bar{v}\}$ and let $(U,V)$ be a partition of $\tilde{v}$. Then the number of pairs $(S,T)$ such that $S\subseteq \bigtimes_{v \in U} \dom(v)$, $T \subseteq \bigtimes_{v \in V} \dom(v)$ and $|S| = |T| = 3\log(N):= a$ is at most $\binom{N}{a}^{\! 2}$. Let $p$ be the probability that there exists such a pair with
		\[\left \{\overrightarrow{\left(\tilde{s} \cup \tilde{t} \right)} \, \mid \, s \in S, \, t \in T \right \} \subseteq R^{\strucB}. \] 
		By the union bound and the definition of $\strucB$ we obtain
		\[
		p \le \binom{N}{a}^{\! \! 2} 2^{-a^2} \le N^{2a}2^{-a^2} = 2^{2a\log(N)-a^2} = 2^{-3\log^2(N)}.
		\]
		Moreover, the number of partitions of $\tilde{v}$ is only a function of $r$ which in particular is at most $|\vertexH|$. Since this number is independent of $N$ the claim follows. 
	\end{claimproof}
	
	\noindent As $|\sigma| = \|A\|$ is independent of $N$ the probability that (\ref{nonuni}) holds for $\strucB$ tends to one as $N \to \infty$. Further, $\|\strucB\| \le |\sigma| \cdot N = \|A\| \cdot N$ and $\strucB$ is coordinate and order respecting relative to $\strucA$ by construction. 
	
	It remains to analyse  the probability that (\ref{big_instance}) holds in $\strucB$. So let $X = |\!\Hom(\strucA, \strucB)|$.  Consider a map $h \colon \vertexH \to B$, where $B$ is the universe of $\strucB$, such that $h(v) \in \dom(v)$ for all $v \in \vertexH$. Then the probability that $h \in \Hom(\strucA, \strucB)$ is exactly the probability that for every $R \in \sigma$ with $R^{\strucA} = \left\{\bar{t} \right\}$, $h(\bar{t}) \in R^{\strucB}$. This occurs with probability $2^{-|\sigma|} = 2^{-\|A\|}:= \delta$. Since the total number of such maps $h$ is equal to $\prod_{v \in \vertexH} |\!\dom(v)| = N^{t}$, we obtain that $\mathbb{E}[X] = \delta N^t$. Let $P$ be the probability that $X \ge \mathbb{E}[X]/8$. Then by a straightforward calculation we obtain.
	\[
	\delta N^t = \mathbb{E}[X] \le (1-P) \cdot \frac{\mathbb{E}[x]}{8} + P \cdot N^t.
	\]
	Rearranging, it follows that $P \ge \frac{7\delta}{8-\delta}$ so since $\delta = 2^{-\|A\|}$ it follows that $P = 2^{-\Omega(\|A\|)}$. This completes the proof since $P$ is positive and independent of $N$. 
	\end{proof} %
 }

\subsubsection{Analysis} 

To argue about these instances we---similarly to in the treewidth case---reason about combinatorial rectangles with partitions balanced relative  to the highly connected nature of $\strucA$. To do this we use Lemma~\ref{lem:weightedrect}, for some appropriately chosen~$f$. The natural choice seems to be to set $f = \mu$\confORfull{.}{, recalling that we extend functions defined on $\vertexH$ to functions defined on subsets of $\vertexH$ in the natural way.} However, for technical reasons we instead use a closely related function $\alpha$ which is defined by
\[
\alpha(v) := \begin{cases}
0 & v \not \in W \\
\frac{1}{2}\sum_{j \in [k]\setminus \{i\}} \sum_{\substack{P \in \mathcal{P}_v}} F_{i,j}(P) & v \in K_i \textit{ for some } i \in [k]
\end{cases}
\]

\noindent The idea is twofold. First, we only take vertices of $W$ into account for our balancing since this is the part of $\graphH$ we know is `highly connected'. And second, we weight each vertex $v \in W$ according to the sum of the weights of paths in $F$ for which $v$ is an endpoint. Trivially, $\alpha(v) \le \mu(v)$ for every $v \in V$. We next observe that every $K_i$ has equal weight under $\alpha$.

\begin{observation} \label{ob: clique_weight}
For every $i \in [k]$, $\alpha(K_i) := \delta = \varepsilon/2 \cdot (k-1)$. Moreover, $\alpha(W) = k \cdot \delta$.  
\end{observation}

\noindent So let $(X,Y)$ be an $\alpha$-balanced partition. This means that globally the weight of $W$ under $\alpha$ is (up to a constant factor) evenly split between $X$ and $Y$. However, this is not necessarily the case for each $K_i$. It could, for instance, happen that every $K_i$ is either a subset of $X$ or a subset of $Y$. We first analyse the case where there are many values of $i$ such that the weight of $K_i$ under $\alpha$ is roughly evenly split between $X$ and $Y$. To formalise this we say that $K_i$ is \emph{$\alpha$-balanced} for $(X,Y)$ if 
\[ 
\min( \alpha(K_i \cap X), \alpha(K_i \cap Y)) \ge \frac{\alpha(K_i)}{10}  =  \frac{\delta}{10}.
\] 
\confORfull{}{Note that the multiplicative constant $1/10$ differs from the $1/3$ that appears in the definition of an $\alpha$-balanced partition; this is unfortunately a technical necessity.} It turns out that if a lot of the $K_i$ are balanced then we can argue very similarly to in the proof of Theorem~\ref{thm: tw_struc}\confORfull{, see Appendix~\ref{ass:rest} for details.}{.}

\begin{lemma} \label{lem: easy_part}
Let $(X,Y)$ be an $\alpha$-balanced partition of $\vertexH$ and let $\lambda > 0$\confORfull{}{ be a constant}. Suppose that $K_i$ is $\alpha$-balanced for at least $\lambda \cdot k$ values of $i$. Then there is a constant $c = c(\lambda) >0$, such that every rectangle $\mathcal{R} \subseteq \Hom(\strucA, \strucB)$ with underlying partition $(X,Y)$ has size at most $N^{t - c \sqrt{k}}$.
\end{lemma}

\confORfull{}{ %
\begin{proof}
	Pick any $K_i$ which is $\alpha$-balanced for $(X,Y)$. Then we know there is some edge $e \in \edgeH$ with $K_i \subseteq e$. Therefore, there is some $R \in \sigma$, with $R^{\strucA} = \{\overrightarrow{e}\}$. Recall that every member of $\mathcal{R}$ is a mapping $h\colon \vertexH \to B$. Let,
	\begin{align*}
		U_i &:= \{h(e \cap X) \, \mid \, h \in \mathcal{R} \} \textit{ and}  \\
		V_i &:= \{h(e \cap Y) \, \mid \, h \in \mathcal{R} \}.
	\end{align*}
	
	\begin{claim}\label{claim:rect}
		Let $u \in U_i$ and $v \in V_i$. Then $\overrightarrow{u \cup v} \in R^{\strucB}$. 
	\end{claim}
	
	\begin{claimproof}
		By the definition of $U_i$ and $V_i$ there exists $h_u, h_v \in \mathcal{R} \subseteq \Hom(\strucA, \strucB)$  such that $h_u(e \cap X) = u$ and $h_v(e \cap Y) = v$. Moreover, since $
		\mathcal{R}$ is a rectangle $h:= h_u\!\restriction_X \times h_v\!\restriction_Y \in \mathcal{R}$. Since $h \in \Hom(\strucA, \strucB)$ we know $h(\overrightarrow{e}) \in R^{\strucB}$. But $h(\overrightarrow{e})  = \overrightarrow{u \cup v}$, since $h(e) = u \cup v$ by construction. The claim follows.
	\end{claimproof}
	
	\noindent 
	On the other hand, by Lemma~\ref{lem: bad_hypergraph}(\ref{nonuni}), $R^{\strucB}$ is $N$-scattered. It follows that either $|U_i| \le 3\log(N)$ or $|V_i| \le 3\log(N)$. Let $\mathcal{I}$ be the set of $i \in [k]$ such that $K_i$ is balanced for $(X,Y)$. By assumption $|\mathcal{I}| \ge \lambda \cdot k$. For every such $i$, let $e_i$ be an edge of $\graphH$ such that $K_i \subseteq e_i$. We then define  
	\[
	s_i := \begin{cases}
		e_i \cap X & \textit{ if } |U_i| \le |V_i| \\
		e_i \cap Y & \textit{ otherwise.} 
	\end{cases}
	\]
	Let $S := \bigcup_{i \in \mathcal{I}} s_i $. It follows that 
	\begin{align*}
		\mu(S) \ge \alpha(S) &\ge \sum_{i \in \mathcal{I}} \min[\alpha(K_i \cap X, K_i \cap Y)] \\
		&\ge  \lambda \cdot k \cdot \frac{\delta}{10} \ge \hat{c} \cdot \sqrt{k},
	\end{align*} 
	for some constant $\hat{c}>0$, depending on $\lambda$, since $\delta = \Omega(k^{-1/2})$.

	Moreover, by definition the restriction of $\mathcal{R}$ to $\{s_i\}$ has size at most $3\log(N)$ for every $i \in \mathcal{I}$. Therefore, since $|\!\Hom(\strucA, \strucB)| = \Theta(N^t)$ and as we take $N$ to be sufficiently large, the result follows by applying Observation~\ref{ob:projectub}.
	\end{proof} %
}

\noindent To finish the proof we have to deal with the case where almost every $K_i$ is not $\alpha$-balanced relative to $(X,Y)$. The key is to argue that then there are many paths in $F$ which cross from $X$ into $Y$ (or vice-versa). This is the core of the argument as it is where the flow is deployed. Note that in the extreme case where every $K_i$ is either a subset of $X$ or a subset of $Y$ the argument is simple. For suppose $K_i \subseteq X$ and $K_j \subseteq Y$, then in particular every path with non-zero weight under $F_{i,j}$ must at some point cross from $X$ into $Y$. Moreover, since $\alpha(K_i) = \delta$ for every $i$ and since the partition $(X,Y)$ is $\alpha$-balanced we know that the number of $K_i$ contained in $X$ differs only by a constant factor from the number of $K_i$ contained in $Y$. Therefore, the total weight of the paths crossing from $X$ into $Y$ is $\Theta(k \cdot \delta) = \Theta(k^{1/2})$; the result then follows by a similar calculation to in Lemma~\ref{lem: easy_part}. In the more general case, where $K_i$ can intersect both $X$ and $Y$, a similar (if slightly more involved) argument works. \confORfull{We give the details in Appendix~\ref{ass:rest}}{Observe that, by Lemma~\ref{lem: easy_part}, we may assume that the weight of many $K_i$ is overwhelmingly concentrated either in $X$ or in $Y$. So we just have to argue that---as in the proof sketch above---the weight of paths intersecting both partitions is large.}

\begin{lemma} \label{lem:hardpart}
Let $\lambda > 0$ be sufficiently small and $(X,Y)$ be an $\alpha$-balanced partition of $\vertexH$, such that the number of $\alpha$-balanced $K_i$ is less than $\lambda \cdot k$. Then there is a constant $d = d(\lambda)>0$ such that every $\mathcal{R} \subseteq \Hom(\strucA, \strucB)$ with underlying partition $(X,Y)$ has size $N^{t - d\sqrt{k}}$.
\end{lemma}
\confORfull{}{
\begin{proof}
	For $Z \in \{X,Y\}$, let $\mathcal{K}_Z$ be the set of $i$ such that $\alpha(K_i \cap Z) < \delta/10 $. Then by assumption 
	\[
	|\mathcal{K}_X \cup \mathcal{K}_Y| \ge k - \lambda \cdot k.
	\]
	We first give a lower bound for $|\mathcal{K}_X|\cdot |\mathcal{K}_Y|$. Observe that
	\[
	\frac{2k\delta}{3} = \frac{2\alpha(W)}{3} \ge \alpha(Y \cap W) \ge |\mathcal{K}_X| \cdot \frac{9 \delta}{10},
	\]
	where the equality uses Observation~\ref{ob: clique_weight}, the first inequality that $(X,Y)$ is $\alpha$-balanced and the second that for every $i \in \mathcal{K}_X$, $\alpha(K_i \cap Y) = \alpha(K_i) - \alpha(K_i \cap X) \ge \delta - \delta/10 = 9\delta/10$. Rearranging, we obtain that $|\mathcal{K}_X| \le  20k/27$. Analogously, we know that $|\mathcal{K}_Y| \le 20k/27$.
	We also know that $|\mathcal{K}_X| + |\mathcal{K}_Y| \ge k - \lambda \cdot k$, so to give a lower bound for $|\mathcal{K}_X|\cdot |\mathcal{K}_Y|$ it suffices to find the minimum of the function 
	\begin{align*}
		&x \mapsto x(k - \lambda \cdot k -x)  & x \in \left[\frac{7k}{27} - \lambda \cdot k, \frac{20k}{27}\right].
	\end{align*} 
	This occurs for the extremal values of $x$, so we obtain that
	\[
	|\mathcal{K}_X|\cdot |\mathcal{K}_Y| \ge \frac{20k}{27} \cdot \left(\frac{7k}{27} - \lambda \cdot k\right) \ge 0.19k^2,
	\]
	where the last inequality follows by choosing $\lambda$ sufficiently small.
	
	Let $F'$ be the union of $F_{i,j}$, such that $i \in \mathcal{K}_X$  and $j \in \mathcal{K}_Y$. Observe that the value of $F'$ is $|\mathcal{K}_X| \cdot |\mathcal{K}_Y| \cdot \varepsilon$. Let $F'_g$ be the subflow of $F'$ containing exactly those paths from $F'$ which intersect both $X$ and $Y$. We aim to show this also has large value. To see this we upper bound the value of $F'_b := F' \setminus F'_g$.
	
	So let $P \in F'_b$, then $P$ is fully contained in $Z$, for some $Z \in \{X,Y\}$. Then there is an endpoint $s(P)$ of $P$ which lies in $K_i$ for some $i \in \mathcal{K}_Z$.   Then note that 
	\[
	\sum_{P \in F'_b} F(P) = 2 \sum_{P \in F'_b} \alpha(s(P)) < 2 (|\mathcal{K}_X|+|\mathcal{K}_Y|) \cdot \frac{\delta}{10} \le \frac{k \cdot \delta}{5} \le \frac{k^2 \cdot \varepsilon}{10},
	\]
	where the last inequality follows since $\delta:= \frac{\varepsilon}{2}\cdot (k-1)$.  It follows that the value of $F'_g$ is at least
	\[
	|\mathcal{K}_X| \cdot |\mathcal{K}_Y| \cdot \varepsilon - \frac{k^2 \cdot \varepsilon}{10} \ge 0.19k^2 \cdot \varepsilon - \frac{k^2 \cdot \varepsilon}{10}  =  \Omega(k^{1/2}),
	\]
	where the asymptotic bound follows since $\varepsilon =\Omega(k^{-3/2})$.
	
	Let $P \in F'_g$. Then we know that at some point $P$ crosses from $X$ to $Y$, i.e. there exist successive vertices  $x_P, y_P$ in $P$ such that $x_P \in X$ and $y_P \in Y$. Therefore, there is some edge $e_P \in \edgeH$ such that $e_P \supseteq \{x,y\}$. It is easy to see by a similar argument to Lemma~\ref{lem: easy_part} that either the restriction of $\mathcal{R}$ to $e_P \cap X$ has size at most $3\log(N)$ or the restriction of $\mathcal{R}$ to $e_P \cap Y$ has size at most $3\log(N)$. Moreover, 
	\[ \mu(e_P \cap X) \ge \mu(x_P) \ge F'_g(P).   \]
	Similarly $\mu(e_P \cap Y) \ge F'_g(P)$. So let $s_P$ equal $e_P \cap X $  if the size of the restriction of $\mathcal{R}$ to $e_P \cap X$ is less than the size of the restriction of $\mathcal{R}$ to $e_P \cap Y$. Otherwise let $s_P$ equal $e_P \cap Y$. Then $\mu(\cup_{P \in F'_g} s_P)$ is at least the value of $F'_g$ which is $\Omega(k^{1/2})$. The result follows by deploying Observation~\ref{ob:projectub}.
\end{proof} %
 }

\noindent Theorem~\ref{thm:sub_lb_reduced}---and therefore also Theorem~\ref{thm: sub_lb}---is immediate.

\begin{proof}[Proof of Theorem~\ref{thm:sub_lb_reduced}]
Take $\strucB$ from Lemma~\ref{lem: bad_hypergraph}. Then $|\!\Hom(\strucA, \strucB)| \ge cN^t$, where $N$, $c$ and $t$ are defined as in the lemma. On the other hand, by Lemmas~\ref{lem: easy_part} and \ref{lem:hardpart} we know that every $\alpha$-balanced rectangle $\mathcal{R}$ with $\mathcal{R} \subseteq \Hom(\strucA, \strucB)$ has size at most $N^{t- \varepsilon\sqrt{k}}$, for some positive constant $\varepsilon$ and some $k = \Omega(\sqrt{w})$; the result follows by Lemma~\ref{lem:weightedrect}.
\end{proof}

        \section{Conclusions}\label{sec:conclusions}
Our main results are both proved using a framework which combines tools from communication complexity with structural (hyper)graph theory. Theorem~\ref{thm: tw_struc}---an unconditional  $N^{\Omega(\tw(\strucA))}$ lower bound on $\{\cup, \times\}$-$\CSP(\classA, \classall)$ for bounded arity $\classA$---shows that this framework can yield tight lower bounds. Moreover, our methods yield a much simpler proof than that of \confORfull{the}{our previous} (inferior) $N^{\tw(\strucA)^{\varepsilon}}$ lower bound \confORfull{of Berkholz and Vinall-Smeeth}{}~\cite{DBLP:conf/icalp/BerkholzV23}. Moreover, Theorem~\ref{thm: sub_lb}---an exact characterisation of when $\{\cup, \times\}$-$\CSP(\classA, \classall)$ is FPT-time solvable---shows that our methods are robust enough to tackle more challenging (higher arity) settings.

This opens up a number of future directions. Firstly, there is the
obvious task of closing the gap between the lower bounds of
Theorem~\ref{thm: sub_lb} and the upper bounds of Theorem~\ref{thm:
subw}. Secondly, can we use our methods to shed light on the
polynomial time solvability of $\{\cup, \times\}$-$\CSP(\classA,
\classall)$? We know this is possible if $\classA$ has bounded
fractional hypertree width and impossible if $\classA$ has bounded
submodular width: can we give an exact characterisation?\confORfull{}{ An answer to
this question could be a step towards establishing a (conditional)
dichotomy on which join queries admit constant delay
enumeration.} Thirdly, it would be interesting to study this problem
for other representation formats.
\confORfull{}{A concrete problem in this direction is to investigate for what
classes $\classA$ can we always produce a \emph{deterministic} $\{\cup,
\times\}$-circuit for $\Hom(\strucA, \strucB)$ in FPT time, which can
be used for exact counting. This is a particularly intriguing problem,
because no conditional dichotomy for the counting variant of
$\CSP(\classA, \classall)$ is known, although there have been some recent
partial results in this direction
\cite{DBLP:journals/tods/KhamisCMNNOS20,
DBLP:journals/corr/abs-2412-06189}.}
        \bibliography{../extras/biblio}
        
               \newpage
               \appendix

\section{On The Relation To Semiring Circuits} \label{a:algebraic}

In this section, we clarify the relationship between our work and a line of work which investigates the smallest circuit computing the \emph{sum-product polynomial} of a conjunctive query \cite{DBLP:conf/icalp/Komarath0R22,DBLP:journals/pacmmod/FanKZ24,DBLP:journals/pacmmod/FanKZ25,DBLP:conf/mfcs/BhargavCC025}. To do this we need some definitions, which we essentially take from \cite{DBLP:journals/pacmmod/FanKZ24,DBLP:journals/pacmmod/FanKZ25} but adapted to our notation.

Given a semiring $\mathbb{S} = (S, \oplus, 
\otimes, \bf{0}, \bf{1})$, a signature $\sigma$ and two $\sigma$-structures $\strucA, \strucB$, we define the \emph{sum-product polynomial} of $\mathbb{S}, \strucA, \strucB$ as:
\[
p_{\strucB}^{\strucA} := \bigoplus_{h \in \Hom(\strucA, \strucB)} \bigotimes_{
\substack{\bar{x}  \in R^{\strucA}  \\  R  \in \sigma }} x_{h(\bar{x})}^R. 
\]
Here each $x_{h(\bar{x})}^R$ is a distinct variable which will take values in $S$, the domain of the semiring $\mathbb{S}$. If we assign each variable to a semiring element then the polynomial evaluates to an element of the semiring. 

The aim of this line of work is it to investigate the size of the smallest circuit over $\mathbb{S}$ which computes a given sum-product polynomial. Formally, a circuit over $\mathbb{S}$ is a fan-in two labelled DAG with a unique sink (output gate) whose internal nodes are labelled by $\oplus$ or $\otimes$ and whose leaves (inputs) are labeled by a variables of the form  $x_{h(\bar{x})}^R$ or one of the constants $\bf{0}$, $\bf{1}$. Moreover, such a circuit is \emph{multilinear} if for every $\otimes$-gate, the two child subcircuits are defined on disjoint sets of variables. Conceptually, this is similar to our notion of decomposability. The semantics of the circuits defined in the obvious way; we refer to \cite{DBLP:journals/pacmmod/FanKZ24,DBLP:journals/pacmmod/FanKZ25} for this and other formal details. If a circuit agrees with a polynomial on every assignment to its variables we say that the circuit \emph{computes} the polynomial. 

Note that we may also regard a circuit as a syntactic object and associate with each circuit $C$ a unique polynomial $p$ formed by `expanding out' the circuit.  We say that $p$ is \emph{produced} by $C$. Note that if a circuit produces a polynomial it also computes this polynomial but not necessarily vice-versa. 

Various upper and lower bounds on the size of these circuits have been acquired in the 
aforementioned literature. For example, Komarath, Pandey and Rahul showed that the smallest circuit computing
$p_{\graphK_n}^{\graphG}$ over the counting semiring has size  $\Theta(n^{\tw(\graphG)+1})$ \cite{DBLP:conf/icalp/Komarath0R22}. This is perhaps reminiscent of our Theorem~\ref{thm: tw_struc}, however the models 
used are fundamentally different.

A crucial conceptual difference is that semiring circuits are
monotone in the following sense: from a semiring circuit computing
$p_{\strucB}^{\strucA}$ one can obtain a semiring circuit
computing $p_{\strucB'}^{\strucA}$ for every substructure $\strucB'\subseteq
\strucB$ by replacing the variables
$x^{R}_t$ for every missing tuple $t\in
R^{\strucB}\setminus R^{\strucB'}$ by $\bf{0}$. This monotonicity does not hold for {\ctc}s:
while there is a small \ctc{} for $\Hom(\graphK_k, \graphK_n)$ this
cannot be transformed to a small circuit representing $\Hom(\graphK_k,
\graphH)$ for every
$\graphH\subseteq \graphK_n$ (by Theorem~\ref{thm: tw_struc}).
One consequence of this phenomenon is that the `hard-instances' used
in
\cite{DBLP:conf/icalp/Komarath0R22,DBLP:journals/pacmmod/FanKZ24,DBLP:journals/pacmmod/FanKZ25}
for lower bound arguments corresponds to easy instances in our model. 

As a result of this lower bounds on the size of semiring circuits computing $p_{\strucB}^{\strucA}$ do not in general imply lower bounds on the size of {\ctc}s computing $\Hom(\strucA, \strucB)$. The next lemma shows this formally for circuits over the Boolean semiring $\mathbb{B} = (\{0,1\}, \vee, \wedge, \bot, \top)$.

\begin{lemma}\label{lem:separation}
        Let $\epsilon>0$ and $k$ be an integer. Then for any sufficiently large integer $n$, there exists a $k$-vertex graph $\graphG$ and a graph $\graphH$ with at least $n$-vertices such that:
        \begin{enumerate}
                \item any circuit computing $p_{\graphH}^{\graphG}$ over the Boolean semiring  has size $\Omega(n^{k-\varepsilon})$, and
                \item there exists a {\ctc} of size $\Theta(n)$ computing $\Hom(\graphG, \graphH)$.
        \end{enumerate}
\end{lemma}     

\begin{proof}
        Let $\graphG = \graphK_k$ and  $\graphH$ be the complete $k$-partite graph with partitions of size $n$. That is the vertex set of $\graphH$ is the disjoint union of $k$ sets of vertices $V_1, \dots, V_k$, each of size $n$, and $\{u,v\} \in \edgeH$ iff $u \in V_i$, $v \in V_j$ for some $i \neq j$. Then Fan, Koutris and Zhao~\cite[Lemma 3.2]{DBLP:journals/pacmmod/FanKZ25} showed that there is no circuit computing $p_{\graphH}^{\graphG}$ over the Boolean semiring of size $O(n^{k-\varepsilon})$. This completes the proof of (1).

        For (2) intuitively we just take a product of all the $V_i$. Formally, the following {\ctc} computes $\Hom(\graphG, \graphH)$ where $S_k$ denotes the set of permutations on $[k]$:
        \[
        \bigcup_{\sigma \in S_k} \bigtimes_{i=1}^k 
        \left( \bigcup_{v \in V_i} x_{\sigma(i)} \mapsto v  \right).
        \]
        The result follows.
\end{proof}

\noindent Moreover, it is well-known that lower bounds on circuits over the Boolean semiring imply lower bounds on circuits over any semiring of characteristic zero \cite[Lemma 11]{DBLP:journals/mst/Jukna15}. Therefore, an analogue of Lemma~\ref{lem:separation} holds for any semiring of characteristic zero.

 Interestingly, in the other direction lower bounds on the size of {\ctc}s computing $\Hom(\strucA, \strucB)$ imply lower bounds on the size of multilinear circuits computing $p_{\strucB}^{\strucA}$ over the Boolean semiring.

\begin{lemma}\label{lem:simulation}
There is a fixed polynomial $q$ such that for every pair of relational structures $(\strucA, \strucB)$ and any multilinear circuit $C$ over the Boolean semiring computing $p_{\strucB}^{\strucA}$, there exists a {\ctc} $\C$ computing $\Hom(\strucA, \strucB)$ of size at most $q(|C|)$.
\end{lemma}      

\begin{proof}
We can propagate out constants from $C$ in a standard way, so we assume that no inputs to $C$ are constants. Further, since $p_{\strucB}^{\strucA}$ is homogeneous and $C$ is multilinear by \cite[Theorem 9]{DBLP:journals/mst/Jukna15}, we may assume that $C$ produces $p_{\strucB}^{\strucA}$. From this the following simple but important claim follows. 

\begin{claim} \label{claim:weak_dec}
Let $g$ be a $\wedge$-gate in $C$ and let $c_1, c_2$ be the two children of $g$. Let $x_{h(\bar{x})}^R$ be a variable occurring in the subcircuit rooted at $c_1$ and $x_{h'(\bar{y})}^S$ be a variable occurring in the subcircuit rooted at $c_2$. Further let $z$ be an element occurring as a coordinate of $\bar{x}$ and $\bar{y}$. Then $h(z) = h'(z)$.
\end{claim}

\noindent Now we show how to convert $C$ into a {\ctc} computing $\Hom(\strucA, \strucB)$. First relabel every $\vee$-gate by $\cup$ and every $\wedge$-gate by $\times$. Every input gate is of the form $x_{h(\bar{x})}^R$. We replace this by a $\times$ gate with inputs labelled by $x \mapsto h(x)$ for every $x \in \bar{x}$. This creates a $\{\cup,
        \times\}$-labelled {\upshape DAG} which we call $\C$. Moreover, it is easy to see that  $\C$ is smooth (i.e. it satisfies condition (2) from Definition~\ref{def:circuit}). However it may not satisfy decomposability  (condition (1) from Definition~\ref{def:circuit}).
        
        This is where Claim~\ref{claim:weak_dec} comes in. Due to the claim $\C$ is \emph{weakly decomposable}. This is a notion from the knowledge compilation literature. Formally a $\times$-gate $g$ with children $g_1, g_2$ is weakly decomposable if for every variable $x$ whenever there is a node in the sub-DAG rooted at $g_1$ with a label of the form $x \mapsto d_1$ and a node  in the sub-DAG rooted at $g_2$ with a label of the form $x \mapsto d_2$ then $d_1 = d_2$. Since it is known \cite[Lemma 80]{de2022hard} that any smooth, weakly decomposable $\{\cup,
                \times\}$-labelled {\upshape DAG} can be converted into an equivalent $\{\cup,
                        \times\}$-circuit with only a polynomial size blow-up, the result follows. 
\end{proof}
\noindent This means that lower bounds on {\ctc}s computing sets of homomorphisms imply lower bounds on the size of multilinear circuits computing sum-product polynomials over the Boolean semiring. In fact, by \cite[Theorem 9]{DBLP:journals/mst/Jukna15} if $p_{\strucB}^{\strucA}$ is multilinear (for instance if it corresponds to computing the answer to a self-join free conjunctive query) we also obtain lower bounds in the counting and tropical semirings.

\end{document}